\documentclass[12pt]{article}
\usepackage[T1]{fontenc}
\usepackage[utf8]{inputenc}
\usepackage[english]{babel}

\usepackage{hyperref}
\hypersetup{colorlinks=true, citecolor=blue, linkcolor=red, filecolor=blue}

\usepackage{amsthm}
\newtheorem{thm}{Theorem}

\newtheorem{prop}{Proposition}

\newtheorem{lem}{Lemma}

\usepackage{url}
\usepackage{amsmath,amsfonts,amssymb}

\usepackage{dutchcal}
\usepackage{enumerate}
\usepackage{xcolor}
\usepackage{algorithm, algorithmicx, algpseudocode}
\usepackage{graphicx, caption, subcaption}
\usepackage{booktabs, tabulary, multirow}
\usepackage{comment}

\usepackage{lscape}
\usepackage{setspace}
\usepackage{adjustbox}
\usepackage[normalem]{ulem}

\usepackage{authblk}
\usepackage{natbib}
\usepackage[toc,page,header]{appendix}

\newcommand{\bu}{\mbox{\bf u}}  
\newcommand{\bD}{\mbox{\bf D}}  
\newcommand{\bI}{\mbox{\bf I}}  
\newcommand{\bU}{\mbox{\bf U}}  
\newcommand{\bZ}{\mbox{\bf Z}}  

\newcommand{\balpha}{\mbox{\boldmath $\alpha$}}

\newcommand{\bDelta}{\mbox{\boldmath $\Delta$}}
\newcommand{\bbeta}{\mbox{\boldmath $\beta$}}
\newcommand{\btheta}{\mbox{\boldmath $\theta$}}
\newcommand{\bkappa}{\mbox{\boldmath $\kappa$}}

\newcommand{\bgamma}{\mbox{\boldmath $\gamma$}}

\newcommand{\bxi}{\mbox{\boldmath $\xi$}}
\newcommand{\bleta}{\mbox{\boldmath $\eta$}}
\newcommand{\bmu}{\mbox{\boldmath $\mu$}}
\newcommand{\bSig}{\mbox{\boldmath $\Sigma$}}

\newcommand{\diag}{\mathrm{diag}} 

\newcommand{\wtbbeta}{\widetilde\bbeta}

\newcommand{\wtbalpha}{\widetilde\balpha}

\newcommand{\talpha}{\tilde\alpha}

\newcommand{\ttheta}{\tilde\theta}

\def\beqn{\begin{eqnarray}}
\def\eeqn{\end{eqnarray}}
\def\beqns{\begin{eqnarray*}}
\def\eeqns{\end{eqnarray*}}

\def\0{{\bf 0}}
\def\A{{\bf A}}

\def\b{{\bf b}}

\def\D{{\bf D}}
\def\d{{\bf d}}

\def\I{{\bf I}}

\def\bO{{\bf O}}

\def\s{{\bf s}}

\def\W{{\bf W}}
\def\w{{\bf w}}
\def\X{{\bf X}}
\def\x{{\bf x}}

\def\y{{\bf y}}
\def\Z{{\bf Z}}

\def\1{{\bf 1}}

\def\trans{^{\rm T}}

\newcommand{\beq}{\begin{equation}}
\newcommand{\eeq}{\end{equation}}

{}

\newcommand{\bes}{\begin{eqnarray*}}
\newcommand{\ees}{\end{eqnarray*}}
\newcommand{\bi}{\begin{itemize}}
\newcommand{\ei}{\end{itemize}}

\newcommand{\bSigma}{\boldsymbol{\Sigma}} 
\newcommand{\bLambda}{\boldsymbol{\Lambda}}

\newcommand{\bOmega}{\boldsymbol{\Omega}}
\newcommand{\nn}{\nonumber}

\newcommand{\bpi}{\boldsymbol{\pi}}

\newcommand{\bphi}{\mbox{\boldmath $\phi$}}

\newsavebox{\tablebox}

\def\bu{\boldsymbol u}

\def\bfA{\mathbf A}

\def\bfI{\mathbf I}

\def\bfZ{\mathbf Z}

\def\bfalpha{\boldsymbol \alpha}
\def\bfbeta{\boldsymbol \beta}

\def\bfgamma{\boldsymbol \gamma}

\def\bftheta{\boldsymbol \theta}

\def\bfeta{\boldsymbol\eta}

\def\bfzero{\boldsymbol 0}

\def\diag{\mbox{diag}}

\def\bSig\mathbf{\Sigma}

\title{Bayesian thresholded modeling for integrating brain node and network predictors}
\author[1]{Zhe Sun}
\author[1]{Wanwan Xu}
\author[2]{Tianxi Li}
\author[3]{Jian Kang}
\author[4]{Gregorio Alanis-Lobato}
\author[1]{Yize Zhao\thanks{ Correspondence author. Email: yize.zhao@yale.edu}}
\affil[1]{Department of Biostatistics, Yale University, New Haven, CT, USA}
\affil[2]{School of Statistics, University of Minnesota, Minneapolis, MN, USA}
\affil[3]{Department of Biostatistics, University of Michigan, Ann Arbor, MI, USA}
\affil[4]{Global Computational Biology and Data Sciences, Boehringer Ingelheim Pharma GmbH \& Co. KG, Biberach an der Riss, Germany}
\date{}

\begin{document}

\markboth%
{Z. Sun and others}
{Thresholded modeling for node and network}

\maketitle


\label{firstpage}

\begin{abstract}
\singlespacing
Progress in neuroscience has provided unprecedented opportunities to advance our understanding of brain alterations and their correspondence to phenotypic profiles. With data collected from various imaging techniques, studies have integrated different types of information ranging from brain structure, function, or metabolism. More recently, an emerging way to categorize imaging traits is through a metric hierarchy, including localized node-level measurements and interactive network-level metrics. However, limited research has been conducted to integrate these different hierarchies and achieve a better understanding of the neurobiological mechanisms and communications. In this work, we address this literature gap by proposing a Bayesian regression model under both vector-variate and matrix-variate predictors. To characterize the interplay between different predicting components, we propose a set of biologically plausible prior models centered on an innovative joint thresholded prior. This captures the coupling and grouping effect of signal patterns, as well as their spatial contiguity across brain anatomy. By developing a posterior inference, we can identify and quantify the uncertainty of signaling node- and network-level neuromarkers, as well as their predictive mechanism for phenotypic outcomes. Through extensive simulations, we demonstrate that our proposed method outperforms the alternative approaches substantially in both out-of-sample prediction and feature selection. By implementing the model to study children's general mental abilities, we establish a powerful predictive mechanism based on the identified task contrast traits and resting-state sub-networks.
\end{abstract}


\clearpage
\doublespacing

\section{Introduction}\label{sec:intro}
It is well known that brain activity is closely linked to human behavior. Thanks to the latest growth in neuroimaging technologies, we are now able to better characterize brain structural and functional alterations from different aspects and hierarchies \citep{zhu2022statistical}. This enables us to investigate brain-to-behavior correspondence and the underlying physiological bases of brain mechanisms associated with behavior and mental processes.

With medical imaging techniques becoming more affordable, most brain imaging studies have started to simultaneously collect different types of images for each participant. By integrating these multiple types of imaging traits, we can better capture changes in the brain and learn about their consequences. Currently, such procedures have mainly been viewed as multi-modality integration problems, for instance, combining measurements from structural magnetic resonance imaging (MRI) and positron emission tomography (PET) to capture both brain structural and functional information. In general, all these modalities are likely summarized by the same type of variate, and further linked together by dissecting their relationships \citep{zhao2021multimodal,zhao2022bayesian} or shared latent structures \citep{zhou2020multi,ma2022multimodal}. On the other hand, most commonly adopted imaging traits can be roughly categorized into two types of variates--measurements at isolated brain locations at voxels or regions of interest (ROIs) referred to as ``nodes'', and patterns of communications across nodes over the brain along edges or connectivity considered as ``networks''.
Given that the information offered by these two types of components is potentially complementary and collected under distinct levels/hierarchies, i.e., node-level and edge-level, it is of great interest to jointly associate both regional measures and network imaging traits with phenotypic outcomes to build more powerful imaging-based predictive mechanisms.

From an analytical perspective, we can summarize the node-level measurements at $P$ nodes by a vector on $\mathbb{R}^{P\times 1}$, with each element measuring the imaging trait at individual locations. This allows us to adopt a variety of high-dimensional regression methods with penalized likelihood or Bayesian variable selection/shrinkage priors to predict behavior \citep{fan2010selective,o2009review}, and some refined versions to further accommodate the spatial correlation among imaging traits \citep{wang2017generalized,goldsmith2014smooth,kang2018scalar,10.1214/15-AOAS818,feng2020bayesian}. For the whole brain network or connectivity, under the same brain atlas or voxels, they can be represented by a symmetric matrix on $\mathbb{R}^{P\times P}$, with each element describing the connection between a pair of nodes. By extracting the unique elements from the matrix, most of the existing methods try to link network predictors with behavior by simplifying the matrix-variate to a single metric or vector \citep{cohen2016segregation,shen2017using, gao2019combining}. However, such simplification techniques often overlook the underlying topological structure of the brain network. Alternatively, by directly taking the matrix as input, low-rank decompositions or graphical representations have been adopted to capture the latent network structure and reduce the feature space \citep{zhou2013tensor,https://doi.org/10.1002/sim.9488,li2018tucker}. Although the two types of imaging features summarized by vector- and matrix-variate offer distinct neurobiological information, little attention has been given yet to integrating them under a joint framework.

In this work, we propose an innovative scalar-on-image regression model based on both vector-variate and matrix-variate neuroimaging traits. Under a unified Bayesian framework, we simultaneously identify signaling node-level and sub-network features accommodating the spatial contiguity across brain anatomy and the biological correlation between different feature components. Specifically, we anticipate an interesting and prevalent interaction between these two types of traits. The most explicit examples are integrating node-level structural imaging with structural connectivity, and node-level functional imaging with functional connectivity. Clearly, the size and shape of a brain location captured by structural MRI (sMRI) would be associated with the anatomical white matter fiber pathways (structural connection) through this location  measured by diffusion MRI (dMRI); and the functional alteration of a brain node measured by functional MRI (fMRI) or PET would also link with the correlation of fMRI time course (functional connection) generated under this node. In other words, there are complex correlations among imaging predictors that could induce coherence in their impact on the outcome. To the best of our knowledge, the closest problem setting to the current one is regression modeling with grouped or structured covariates.  Under frequentist paradigms, different penalty functions have been proposed to encourage similarity in the effect size or significant status
if two covariates are in the same group or correlated due to similar structures
\citep{simon2013sparse,li2008network,zhao2016hierarchical}. In Bayesian models, popular prior choices including a combination of spike and slab prior with binary Markov random field prior \citep{li2010bayesian,stingo2011incorporating},  and graphical prior embedded with shrinkage prior \citep{chang2018scalable} can also achieve similar goals.
However, unlike existing scenarios, the correlations among covariates in our problem arise from features present on both the nodes and edges connecting these nodes in a network.

This work makes several important contributions to the literature. First, we propose a novel analytical framework that integrates vector- and matrix-variate predictors to model their joint effect on the outcome. This approach fills a critical gap in neuroscience research by enabling the study of hierarchical brain alterations and their impact on phenotypic outcomes. Second, we develop an innovative joint thresholded prior model extending the existing thresholded prior literature \citep{ni2019bayesian,Cai2020,wu2022bayesian} to identify informative node- and network-level features, as well as capture the interplay and coupling effect of their signal patterns.  Instead of modeling individual edges, we also explicitly account for the topological architectures of the network components and uncover sub-network structures to improve biological plausibility and interpretability. Finally, we implement our method to the latest landmark children's brain developmental study and achieve a substantially improved predictive performance for general mental abilities based on brain task contrast maps and resting-state connectivity. Our results highlight the potential of integrating different levels of neuronal activities to inform human behavior profiles.

The rest of the article is organized as follows. In Section \ref{sec:model}, we introduce the model formulation under vector- and matrix-variate predictors and develop joint thresholded priors to accommodate their interplay. In Section \ref{sec:comp}, we provide the posterior inference. We demonstrate the performance of our method compared with existing ones and our own variations via simulations in Section \ref{sec:simu}, followed by the application to the Adolescent Brain Cognitive Development (ABCD) study in Section \ref{sec::app}. We conclude our paper with a discussion in Section \ref{sec:dis}.

\section{Model} \label{sec:model}
\subsection{Modeling setup}
Suppose there are $N$ subjects in a brain imaging study. For subject $i (i=1,\dots,N)$, based on the collected brain imaging data, multiple types of neuroimaging traits are summarized under a brain atlas with $P$ regions of interest (ROIs) or nodes. These include a regional metric summarized across all the nodes denoted by a vector $\x_i \in \mathbb{R}^{P\times 1}$; and a whole brain network measure, i.e., brain connectivity, represented by a symmetric matrix $\Z_i\in \mathbb{R}^{P\times P}$ with each of its off-diagonal element $z_{ikl}\in \mathbb{R}$ characterizing the connection between nodes $k$ and $l$ with $k\neq l$, and we set $z_{ikl}=0$ when $k=l$. Our goal is to study how the regional and network neuromarkers impact a continuous phenotypic outcome of interest denoted by $y_i$ in an interactive way, while adjusting for non-imaging covariates and clinical confounders summarized by $\w_i \in \mathbb{R}^{Q \times 1}$ with the first element being one to accommodate the intercept. As mentioned previously,
given the neurobiological dependence between features collected on brain nodes and the edges connecting these nodes, for instance, regional volume and structural connectivity, or activation map and functional connectivity, it is important to assume certain analytical correlation between the parameters associated with node-level features and the features along the connecting edges during the learning process.
To achieve this, we propose the following 
linear model with both vector-variate and matrix-variate predictors
\begin{align}
    y_i& =\bfeta\trans\w_i+\bfbeta\trans\x_i+\langle\A, \Z_i\rangle+\epsilon_i, \label{eq:origial}
\end{align}
where $\epsilon_i \overset{\text{i.i.d}}{\sim} \mbox{N}(0, \sigma_{\epsilon})$; $\bfeta\in \mathbb{R}^{Q\times 1}$ are the coefficients for the non-imaging covariates; $\bfbeta \in \mathbb{R}^{P\times 1}$ are the coefficients characterizing the effect from the node-wise measurements on the outcome; and matrix $\Z_i$ are associated with the outcome by an inner product with a symmetric coefficient matrix $\A=(a_{kl}) \in \mathbb{R}^{P \times P}$, given that $\langle\A, \Z_i\rangle=\mbox{vec}(\A)^T\mbox{vec}(\Z_i)$ with $\mbox{vec}(\cdot)$ being the vectorization operation on the matrix. With $a_{kl}=a_{lk}, k\neq l \in \{1,\dots,P\}$, each $a_{kl}$ captures the impact of connection/edge $z_{ikl}$ on the outcome.
Under a large number of unknown parameters in model \eqref{eq:origial}, particularly coming from the coefficient matrix $\A$, it is desirable to perform dimension reduction to reduce the parameter space. Moreover, brain network architectures have been shown to reveal a low-rank structure \citep{wang2019common}. These motivate us to  employ the following rank-$R$ PARAFAC decomposition \citep{kolda2009tensor} under a symmetric constraint
\begin{align} \label{eq:deco}
  \A=\sum_{r=1}^R\bfalpha^{(r)}\otimes\bfalpha^{(r)},
\end{align}
where $\bfalpha^{(r)} \in \mathbb{R}^{P \times 1}, r=1,\dots, R$ are column vectors, and operator ``$\otimes$'' denotes the outer product. Such a decomposition setup has been adopted previously for general tensor predictors in a regression model framework \citep{guhaniyogi2017bayesian,zhou2013tensor}. Under a special case when rank $R=1$, model \eqref{eq:deco} will reduce to a quadratic term with a bilinear version without a symmetric condition focused by \cite{hung2013matrix}. Combining models \eqref{eq:origial} and \eqref{eq:deco}, we can represent our model by
\begin{equation}
  \begin{split} \label{eq:decomposed}
    y_i
    & = \bfeta\trans\w_i+\bfbeta\trans\x_i+\langle\sum_{r=1}^R\bfalpha^{(r)}\otimes\bfalpha^{(r)}, \Z_i\rangle+\epsilon_i \\
    & = \bfeta\trans\w_i+\bfbeta\trans\x_i+\sum_{r=1}^R (\bfalpha^{(r)})\trans \Z_i\bfalpha^{(r)}+\epsilon_i.
  \end{split}
\end{equation}
Model \eqref{eq:decomposed} clearly reflects that the impact on the outcome from brain regional measurements is captured by a coefficient vector $\bfbeta$, and the impact from brain connectivity is characterized by $R$ components with each one represented by $\bfalpha^{(r)}\otimes\bfalpha^{(r)}=\bfalpha^{(r)}(\bfalpha^{(r)})^T\in\mathbb{R}^{P \times P}$. 
As is common practice for high-dimensional data, we first assume sparsity among $\bfbeta$ with the anticipation that only a small proportion of regional measurements in $\x_i$ are associated with the outcome. In terms of the connectivity effect, 
from a topological point of view, when $\bfalpha^{(r)}$ is sparse, the non-zero elements in the corresponding outer product matrix could characterize an informative clique sub-network linked with the outcome \citep{wang2021learning}. This indicates that, besides the regional neuromarkers, we can simultaneously identify the signaling sub-network configurations to help us better understand the neurobiological etiology and inform potential intervention targets for future clinical strategies.

Finally, to ensure that the identified signaling clique sub-networks from model \eqref{eq:decomposed} are uniquely defined, we denote $g_r \subset \{1, 2, \cdots, P\}$ to be the support set of $\bfalpha^{(r)}$, i.e., $g_r = \{j: \alpha^{(r)}_j \ne 0\}$ for $r=1, \cdots, R$. The uniqueness of the clique set $\{g_r\}_{r=1}^R$  can be guaranteed by the following theorem:
\begin{thm}\label{the}
\textit{Suppose $\{g_r\}_{r=1}^R$ is support consistent with coefficient matrix $\bfA$ and each $g_r$ contains at least one unique variable that is not included in other sets. The clique set $\{g_r\}_{r=1}^R$ is the unique and minimum support consistent clique set of matrix $\bfA$.}
\end{thm}
The proof of Theorem \ref{the} is provided in Section A of the Supplementary Materials. Based on this theorem, in real practice, we can straightforwardly verify the uniqueness of the identified signaling sub-networks from the final numerical solutions by checking whether each of the identified clique sets provides at least one unique node index. Such uniqueness on the solution of $\{\bfalpha^{(r)}\}_{r=1}^R$ ensures the plausibility of our interpretation on the model estimations.
\subsection{Thresholded prior for node and network predictors} \label{sec:prior}
To jointly identify signals among vector-variate (node measure) and matrix-variate (connectivity measure) predictors in a biologically plausible way, in this section, we will discuss a set of new thresholded prior models we develop for the imaging-related coefficient components along with prior settings for the other nuisance- and hyper-parameters. As one of the essential components in our prior specifications, the neurobiological relationship between node and connectivity measurements needs to be carefully considered.

To select informative node-specific and sub-network imaging features, we impose sparsity on $\bfbeta$ and each $\bfalpha^{(r)}, r=1,\dots,R$ by first performing the following decomposition
\begin{align}\label{eq:prior1}
  \bfbeta&=\widetilde{\bfbeta} \circ \mathcal{T}(\bfgamma; \lambda),  \qquad
  \bfalpha^{(r)}=\widetilde{\bfalpha}^{(r)} \circ \mathcal{T}(\bftheta, \bftheta^{(r)}; \lambda);\quad r=1,\dots, R.
\end{align}
Here, the operator ``$\circ$'' represents an entry-wise product, parameters $\widetilde{\bfbeta}\in \mathbb{R}^{P \times 1}$ and $\{\widetilde{\bfalpha}^{(r)}\in \mathbb{R}^{P \times 1};  r=1,\dots, R\}$ represent the non-zero effect from the node and network predictors, respectively, and functions $\mathcal{T}(\bfgamma; \lambda)$ and $\mathcal{T}(\bftheta, \bftheta^{(r)}; \lambda)$ induce a distinction of zero and non-zero components in the corresponding effect vector. Based on model \eqref{eq:prior1}, we decompose each coefficient component into its latent non-zero element and a controlling vector function to cartography signal locations. To ensure signals to be spatially contiguous across the brain anatomy as is anticipated in brain imaging applications, we consider the following vector thresholded functions
\begin{equation}
  \begin{aligned}\label{eq:thre}
    \mathcal{T}(\bfgamma; \lambda) &= \Big\{\mathbb{I}(|\gamma_1|>\lambda), \ldots, \mathbb{I}(|\gamma_P|>\lambda)\Big\}\trans, \\
    \mathcal{T}(\bftheta, \bftheta^{(r)}; \lambda) &= \Big\{\mathbb{I}(|\theta_1|>\lambda) \mathbb{I}(|\theta^{(r)}_1|>\lambda), \ldots, \mathbb{I}(|\theta_P|>\lambda) \mathbb{I}(|\theta^{(r)}_P|>\lambda)\Big\}\trans,
  \end{aligned}
\end{equation}
where $\mathbb{I}(\omega)$ is the indicator function with $\mathbb{I}(\omega)=1$ if $\omega$ holds and zero otherwise. Based on model \eqref{eq:thre}, function $\mathcal{T}(\bfgamma; \lambda)$ consists of a latent vector $\bfgamma=(\gamma_1,\dots,\gamma_P)\trans$ with each of its elements continuous over the domain, and a non-negative threshold parameter $\lambda$ to control the degree of sparsity for $\bfbeta$. Function $\mathcal{T}(\bftheta, \bftheta^{(r)}; \lambda)$ contains two latent continuous vectors $\bftheta=(\theta_1,\dots,\theta_P)\trans$ and $\bftheta^{(r)}=(\theta^{(r)}_1,\dots,\theta^{(r)}_P)\trans$ and the same threshold parameter $\lambda$. Specifically,  $\bftheta$ serves as a generic latent component by imposing a global sparsity for the overall graphic architecture across all the $R$ sub-networks; and $\bftheta^{(r)}$ serves as the individualized one to capture the unique signaling network patterns deviating from the global one. Eventually, the sparsity of $\bfalpha^{(r)}$ is determined simultaneously by both global and individualized latent components. Such a prior construction for $\{\bfalpha^{(r)}\in \mathbb{R}^{P \times 1};  r=1,\dots, R\}$  is in a similar spirit to the sparse group selection literature \citep{chen2016bayesian,simon2013sparse} to impose sparsity in a hierarchical way to accommodate individual signals and their joint interplay from a group perspective, which has shown a great power for the structured variable selection.

Another unique advantage of models \eqref{eq:prior1} and \eqref{eq:thre} is that we could explicitly delineate the correlation between the informative neuromarkers within regional and connectivity measurements. Under the expectation that signal locations captured by brain regional metric and those captured by the networks linked with specific regions should have a coupling effect, we impose the following joint Gaussian prior distribution for $\bfgamma$ and the global network component $\bftheta$ as
\begin{align}
  \begin{pmatrix}
    \bfgamma \\
    \bftheta \\
  \end{pmatrix} \sim
  \mbox{N} \left\{\begin{pmatrix} \boldsymbol{0}_P \\
    \boldsymbol{0}_P\\
    \end{pmatrix},
      \begin{pmatrix}
        \bSigma_{\gamma} & \bLambda^T \\
        \bLambda &\bSigma_{\theta}
      \end{pmatrix}
    \right\},
    \label{eq:corr}
\end{align}
Another critical modeling setup to induce the desired coupling effect is to have a shared thresholded parameter $\lambda$ for functions $\mathcal{T}(\bfgamma; \lambda)$ and $\mathcal{T}(\bftheta, \bftheta^{(r)}; \lambda)$. This ensures that individual $\mathbb{I}(|\gamma_p|>\lambda)$ and $\mathbb{I}(|\theta_p|>\lambda)$ are more likely to yield the same result when $\gamma_p$ and $\theta_p$ are correlated.
In terms of matrices $\bSigma_{\gamma}$ and $\bSigma_{\theta}$, we could leave them unstructured and assign them Inverse Wishart priors, or adopt diagonal matrices. In our numerical studies, we decide to assume the covariance matrices to be squared exponential with variance parameters $\sigma_{\gamma}$ and $\sigma_{g}$, respectively, and a shared lengthscale parameter $\nu$ to mimic the Gaussian process \citep{kang2018scalar}.
This could better promote the spatial contiguity of the signals.
In terms of individualized network component $\bftheta^{(r)}$, we assume $\bftheta^{(r)}\sim \mbox{N}(\bftheta, \sigma_{\theta}\bfI_{P})$ for $r=1,\dots, R$, where the global component serves as the mean parameter to align with the overall effect.

For the non-zero effect parameters in model \eqref{eq:prior1},  we assign $\widetilde{\bfbeta}\sim \mbox{N}(\bfgamma,\sigma_{\beta}\bfI_{P})$ and $\widetilde{\bfalpha}^{(r)}\sim \mbox{N}(\bftheta,\sigma_{\alpha}\bfI_{P}), r=1,\dots, R$.  We choose their mean parameters in this way so that the magnitude of the non-zero effect could be close to the latent continuous components to be thresholded. This could ensure that the absolute values of effect size $\bfbeta$ and $\{\bfalpha^{(r)}; r=1,\dots, R\}$ corresponding to the selected signals are large enough to reflect a sufficient impact.
Finally, for the remaining nuisance and hyper-parameters, we set $\bleta \sim \mbox{N}(\0,\sigma_{\eta}\bfI_Q)$ with $\sigma_{\eta}$ pre-specified to a large value, $\lambda$ a non-informative Uniform prior $\mbox{U}(0, \lambda_{max})$ with $\lambda_{max}$ a conservative upper bound,  and assign non-informative Inverse Gamma (IG) priors for the remaining variance parameters. Overall, we name our model \textit{\textbf{J}oint \textbf{N}ode and \textbf{N}etwork \textbf{T}hresholded \textbf{s}election} (JNNTs) to indicate its unique constructions to select both node-level traits and sub-networks in a coherent and topologically plausible framework.
The schematic diagram of JNNTs modeling is summarized in Figure \ref{fig:model}.

\begin{figure}[!htp]
  \centering
  \includegraphics[width=1\linewidth]{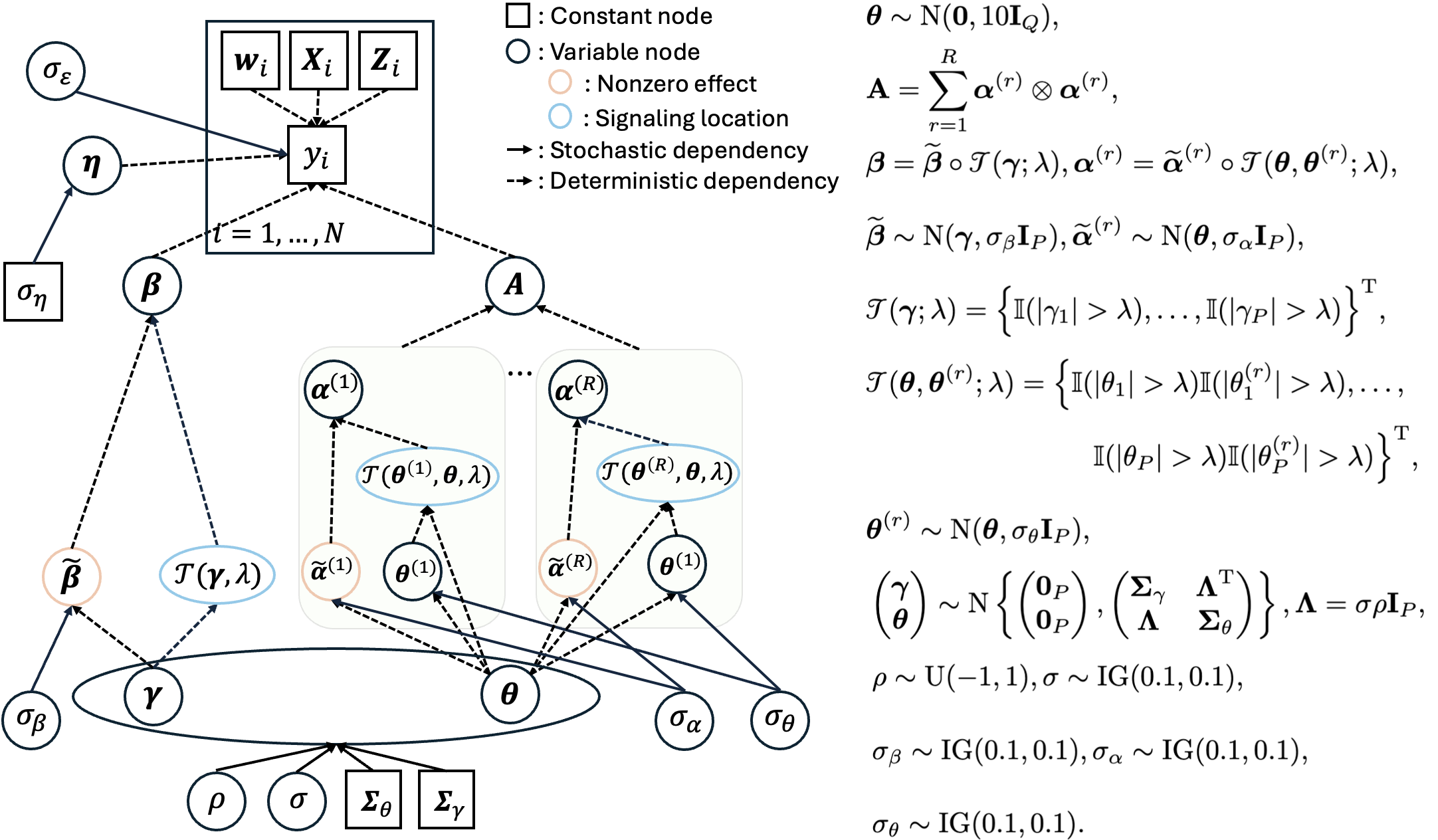}
  \caption{Graphical representation of the JNNTs model.}
  \label{fig:model}
\end{figure}

\section{Posterior Inference} \label{sec:comp}
\subsection{Markov chain Monte Carlo algorithm}
As aforementioned, there are different choices for covariance matrices $\bSigma_{\gamma}$ and $\bSigma_{\theta}$. Here we consider them to be squared exponential to accommodate the spatial smoothness over brain. The posterior algorithm can be easily simplified to the case when covariance matrices are diagonal.  Specifically, we have $\bSigma_{\gamma}=\sigma_{\gamma}\mathbf{O}$ and $\bSigma_{\theta}=\sigma_{g}\mathbf{O}$, where $\mathbf{O}\in \mathbb{R}^{P \times P}$ is a covariance kernel matrix with each of its element $o_{pj}=\exp\{-\|\s_p - \s_j\|_2^2/2\}$, and $\s_p \in \mathbb{R}^{3}$ is the 3D coordinate for node $p$ $(p=1, \ldots, P)$ under the brain registered space. In such specifications, the covariance matrix in the joint distribution \eqref{eq:corr} is positive definite if and only if $|\rho|\sigma_c \le \mbox{min}\big(\{d_p\}_{p=1}^P\big)\sqrt{\sigma_{g}\sigma_{\gamma}}$, where $d_1,\dots,d_p$ are the eigenvalues of $\mathbf{O}$. Given $\sigma_{c}\rho$ are joint together to capture the correlation, we could fix $\sigma_{c}$ while varying $\rho$. Without loss of generality, we assume $\sigma_{\gamma}=\sigma_{g}=\sigma$, and $\sigma_{c} = \sigma/\delta$ with $\delta$ a pre-fixed large number in practice to reduce the number of hyper-parameters. 
We represent the observed data $\mathcal{D}=\{y_i, \w_i, \x_i, \Z_i, i=1,\dots,N\}$.
The posterior likelihood for unknown parameters $\Phi = \big\{\lambda, \bleta, \bgamma, \wtbbeta, \btheta, \{\btheta^{(r)}\}_{r=1}^R, \allowbreak \{\wtbalpha^{(r)}\}_{r=1}^R, \sigma_{\beta}, \sigma_{\alpha}, \sigma_{\theta}, \sigma_{\epsilon}, \sigma, \rho \big\}$ given the data becomes
\begin{align*}
  \bpi(\Phi \mid \mathcal{D}) \propto
  & \bpi(\mathcal{D} \mid \lambda, \bleta, \bgamma, \wtbbeta, \btheta, \{\btheta^{(r)}\}_{r=1}^R, \{\wtbalpha^{(r)}\}_{r=1}^R, \sigma_{\epsilon}) \bpi(\bleta) \bpi(\wtbbeta \mid \bgamma, \sigma_{\beta})  \\
  & \times \prod_{r=1}^R \bpi(\wtbalpha^{(r)} \mid \btheta, \sigma_{\alpha}) \prod_{r=1}^R \bpi(\btheta^{(r)} \mid \btheta, \sigma_{\theta}) \bpi(\bgamma, \btheta \mid \bSigma_{\gamma}, \bSigma_{\theta}, \sigma, \rho) \\
  & \times \bpi(\sigma_{\beta}) \bpi(\sigma_{\alpha})\bpi(\sigma_{\theta})
  \bpi(\sigma_{\epsilon}) \bpi(\sigma) \bpi(\rho) \bpi(\lambda).
\end{align*}
Based on the posterior distribution, we develop a Markov chain Monte Carlo (MCMC) algorithm to conduct the posterior inference with a combination of Gibbs sampler and Metropolis-Hastings (MH) \citep{MH1953} steps .
A brief demonstration of the sampling steps is shown below and a detailed description of the MCMC algorithm is provided in Section B of the Supplementary Materials.

Briefly, during the MCMC, we iteratively
\begin{itemize}
    \item Sample $[\bleta \mid \mathcal{D}, \lambda, \bgamma, \wtbbeta, \btheta, \{\btheta^{(r)}\}_{r=1}^R, \{\wtbalpha^{(r)}\}_{r=1}^R, \sigma_{\epsilon}]$ and $[\wtbbeta \mid \mathcal{D}, \lambda, \bleta, \bgamma, \btheta, \\ \{\btheta^{(r)}\}_{r=1}^R \{\wtbalpha^{(r)}\}_{r=1}^R, \sigma^2_{\epsilon}, \sigma^2_{\beta}]$ from their full conditional Normal distributions, respectively.
    \item Sample $[ \talpha_{p}^{(r)} \mid  \mathcal{D}, \lambda, \bleta, \bgamma, \wtbbeta, \btheta, \{\btheta^{(r)}\}_{r=1}^R, \{\wtbalpha^{(k)}\}_{k \neq r}^P, \{\talpha_j^{(r)}\}_{j \neq p}^P, \sigma_{\epsilon}, \sigma_{\alpha}]$ from the full conditional Normal distributions, for $r=1, \ldots, R, p=1, \ldots, P$.
    \item Sample $[\theta^{(r)}_{p} \mid  \mathcal{D}, \lambda, \bleta, \bgamma, \wtbbeta, \btheta, \{\btheta^{(k)}\}_{k \neq r}, \{\theta^{(r)}_j\}_{j \neq p}, \{\wtbalpha^{(r)}\}_{r=1}^R, \sigma_{\epsilon}, \sigma_{\theta}]$ from a mixed truncated Normal distributions, for $r=1, \ldots, R;  p=1, \ldots, P$.
    \item Sample  $[\gamma_{p} \mid  \mathcal{D}, \lambda, \bleta, \{\gamma_j\}_{j \neq p}^P, \wtbbeta, \btheta, \{\btheta^{(r)}\}_{r=1}^R, \{\wtbalpha^{(r)}\}_{r=1}^R, \sigma_{\epsilon}, \sigma_{\beta}, \sigma, \rho ]$ and $ [\theta_{p} \mid  \mathcal{D}, \lambda, \bleta, \bgamma, \wtbbeta, \{\theta_{j}\}_{j \neq p}, \{\btheta^{(r)}\}_{r=1}^R, \{\wtbalpha^{(r)}\}_{r=1}^R, \sigma_{\epsilon}, \sigma_{\alpha}, \sigma_{\theta}, \sigma, \rho]$ from mixed truncated Normal distributions, respectively, for $p = 1, \ldots, P$.
    \item Sample $[\sigma_{\beta} \mid \wtbbeta, \bgamma]$, $[\sigma_{\alpha} \mid \{\wtbalpha^{(r)}\}_{r=1}^R, \btheta]$, and $[\sigma_{\theta} \mid \{\btheta^{(r)}\}_{r=1}^R, \btheta]$ from the conditional IG distributions, respectively.
    \item Sample $[\sigma_{\epsilon} \mid \mathcal{D}, \lambda, \bleta, \bgamma, \wtbbeta, \btheta, \{\btheta^{(r)}\}_{r=1}^R, \{\wtbalpha^{(r)}\}_{r=1}^R]$ from the conditional IG distribution.
    \item Sample $[\sigma \mid \bgamma, \btheta, \rho]$ from the conditional IG distribution.
    \item Sample the proposed $\rho$ and $\lambda$ based on random walk proposals, then  calculate the acceptance ratios $R_{\rho}$ and $R_{\lambda}$ with probabilities min$(1, R_{\rho})$, min$(1, R_{\lambda})$ to accept the proposal values, respectively.
\end{itemize}
The convergence of the MCMC will be assessed by trace plots and GR method \citep{gelman1992}. Based on the posterior samples, we could determine the informative node-level and network-level features and their associated effects from the posterior means along with uncertainty quantification. Specifically, at each iteration, node $p$ is considered as selected when $\mathbb{I}(|\gamma_p|>\lambda)=1$, and edge $(p,j)$ at sub-network $r$ is selected when  jointly $\mathbb{I}(|\theta_p|>\lambda) \mathbb{I}(|\theta^{(r)}_p|>\lambda)\mathbb{I}(|\theta_j|>\lambda) \mathbb{I}(|\theta^{(r)}_j|>\lambda) =1$.
After calculating the marginal posterior  probabilities (MPP) of those metrics over all the iterations after burn-in, a cutoff of 0.5 on MPP will be used to determine whether they are significantly associated with the outcome according to the median
probability model \citep{hastie2004}. Of note, the current way to summarize the sub-networks is different from calculating the MPP for each individual edge and assembling the ones with MPPs larger than the cutoff, given that the latter one does not reflect the inclusion probability of a network configuration as a whole.
\section{Simulation} \label{sec:simu}
We conduct simulation studies to assess the finite sample performance of the proposed JNNTs model on both prediction, and node- and network-level feature selections.  We consider a low-dimensional atlas with $P=20$ and a high-dimensional one with $P=100$. Under the low-dimensional case, we consider two different sample sizes with $N=50$ or 200; and we set $N=1000$ for the high-dimensional case. These setups will make the simulated data scales similar to or more challenging than our data application.  In terms of signal patterns, we consider four different scenarios as shown in Figure \ref{fig:simulation-circle}  with the signaling nodes colored in orange and signaling sub-networks displayed. For the first scenario generated for $P=20$, we impose a strong coupling on the signals for node and network components with all the nodes in the two signaling sub-networks also generating  node-level signals. In the second scenario for $P=20$, on the contrary, there is no coupling between the signals from the two components with none of the nodes in the signaling sub-networks providing significant node-level information.
In the third scenario, we maintain the same setting as that in the first scenario but remove two edges from the signaling subgraph colored in blue. This setting allows us to assess the robustness of JNNTs when the model assumption does not hold.
The fourth scenario, which is designed for the high-dimensional brain atlas, is positioned as a mixture of the previous scenarios to mimic a more practical setting with a certain correlation between the regional and network signal sources but they are not fully aligned. Based on the signal patterns under each scenario, we generate the non-zero elements of $\bbeta$ from $\mbox{N}(\bgamma, \I_P)$, and those of $\balpha^{(r)}$ from $\mbox{N}(\btheta, 2\I_P)$ for $r=1, 2$, where $\bgamma$ and $\btheta$ are both generated from $\mbox{N}(\bfzero, 2\I_P)$. We then set $\bfeta=1$, and generate $\x_i$ and the off-diagonal elements of $\bfZ_i$ from $\mbox{N}(0, 1)$. Finally, we consider two noise levels with $\sigma_{\epsilon}=2$ for a low-noise setting and 6 for a high-noise one. In total, we consider 14 different simulated settings.

\begin{figure}[!htbp]
  \centering
  \begin{subfigure}{0.48\textwidth}
      \centering
      \includegraphics[width=1\linewidth]{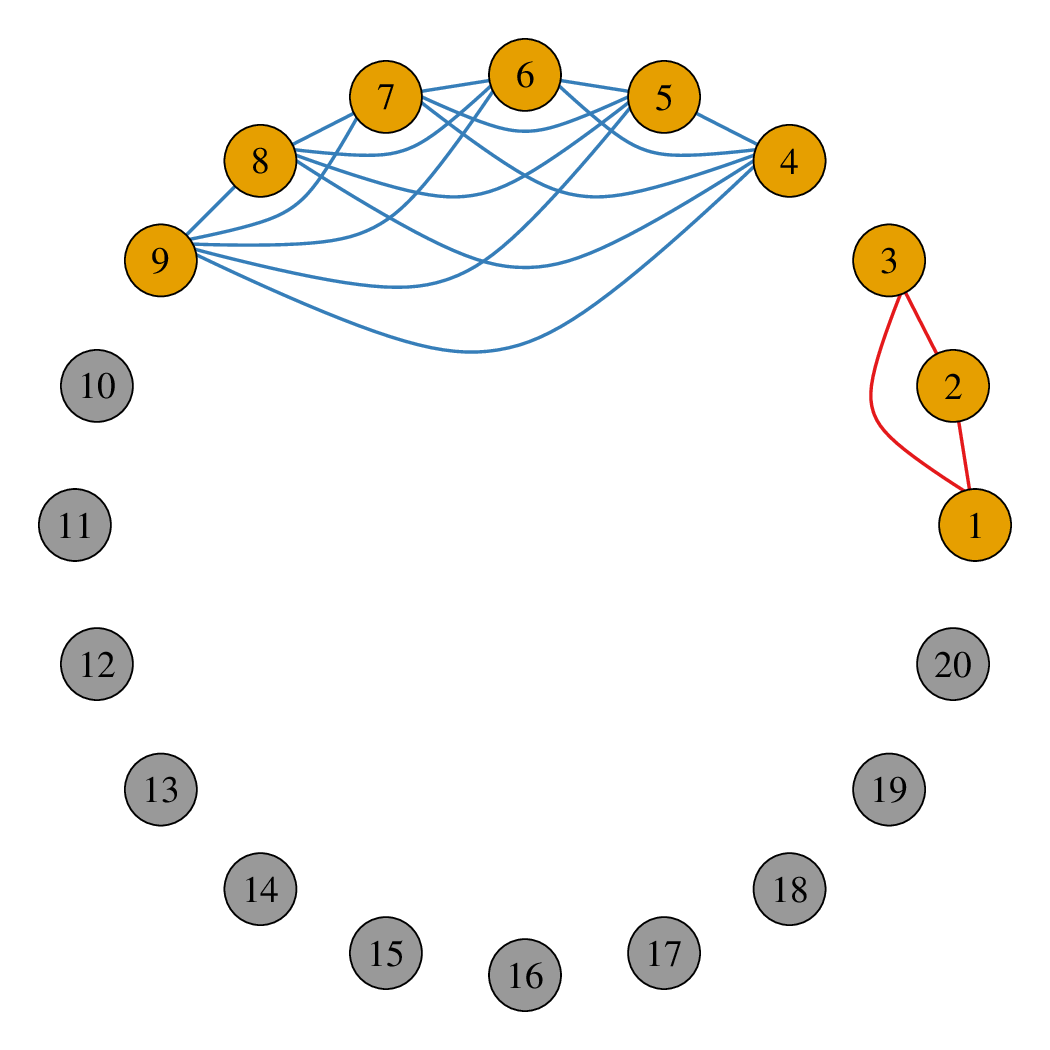}
      \caption{Scenario 1.}
      \label{fig:setting1}
  \end{subfigure}
  \hspace*{0.1cm}
  \begin{subfigure}{0.48\textwidth}
      \centering
      \includegraphics[width=1\linewidth]{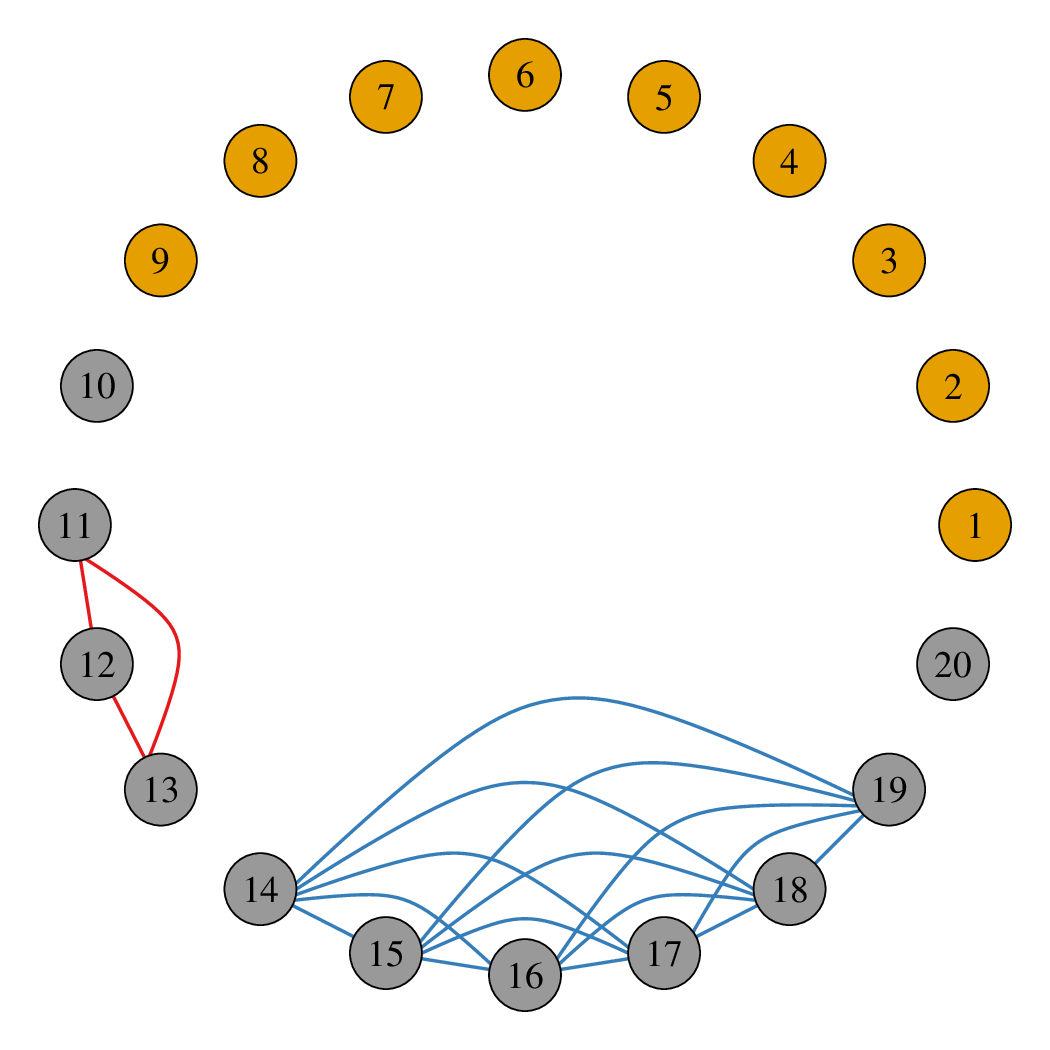}
      \caption{Scenario 2.}
      \label{fig:setting2}
  \end{subfigure}
  \begin{subfigure}{0.48\textwidth}
      \centering
      \includegraphics[width=1\linewidth]{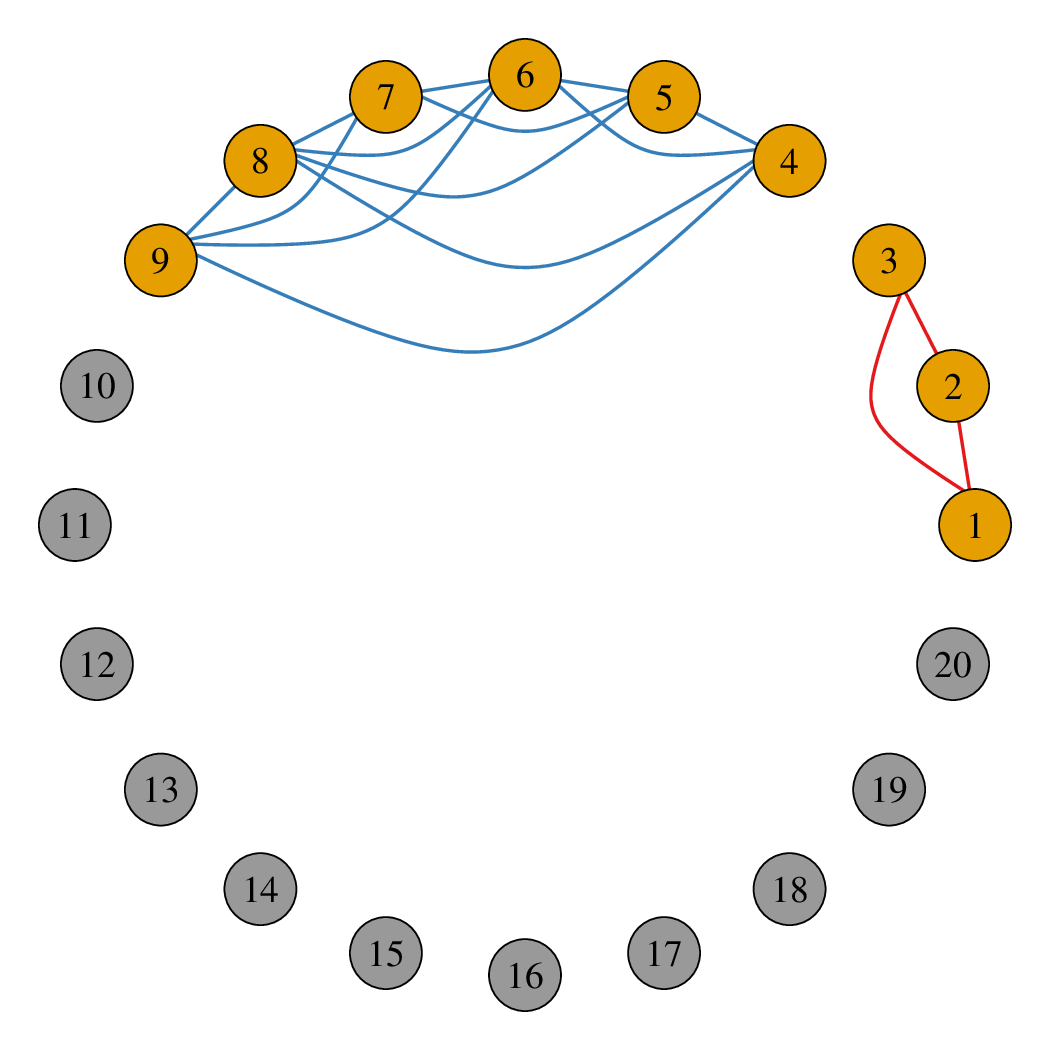}
      \caption{Scenario 3.}
      \label{fig:setting3}
  \end{subfigure}
  \hspace*{0.1cm}
  \begin{subfigure}{0.48\textwidth}
      \centering
      \includegraphics[width=1\linewidth]{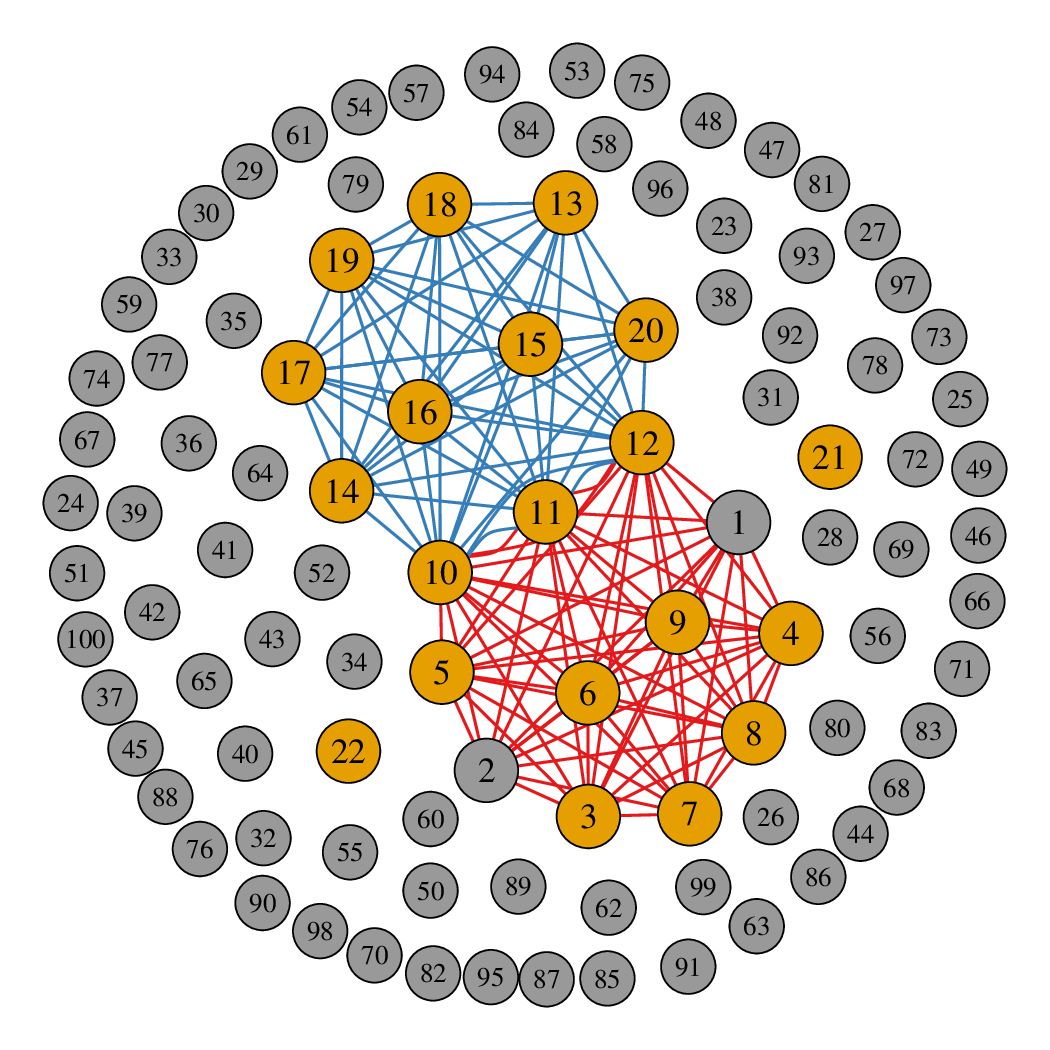}
      \caption{Scenario 4.}
      \label{fig:setting4}
  \end{subfigure}
  \caption{Illustration figures for simulation setups.
  In all figures, nodes colored in orange and grey are assigned with non-zero and zero coefficients, respectively. Sub-networks that are relevant to the prediction of the outcome are indicated with different edge colors.}
  \label{fig:simulation-circle}
\end{figure}

For each simulated dataset, to determine the tuning parameter and evaluate the out-of-sample prediction, we randomly split 10\% of its samples to form a validation set and generate an
independent testing set of size 100. We implement the proposed JNNTs model  with $\lambda_{max}=1.5$, $\sigma_{\eta}=10$, and IG(0.1, 0.1) hyper-prior for the remaining variance parameters. We set the candidate set of $R$ to be $\{2, 3, 4, 5\}$ with the optimal value determined by $R^2$ under the validation set. Based on random initials, each MCMC contains 10,000 iterations with 5,000 burn-in.
Besides JNNTs, we also implement several competing methods. Given that there is no existing method that can simultaneously integrate vector- and matrix-variate predictors, we consider Lasso \citep{LASSO} and Horseshoe \citep{Horseshoe2009} as natural choices to perform high-dimensional regression under a frequentist and Bayesian paradigm. To apply them to our problem, we extract the upper diagonal elements from the network matrix-variate as a vector, and stack this vector with the node-level features as predictors. We also consider a low-rank tensor regression model \citep{zhou2013tensor}, denoted as Tensor, which allows only matrix-variate predictors. Since their estimated coefficient matrix is not symmetric, we sum the matrix with its transpose and divide by 2 as the final estimation.
Besides these existing methods, we further consider two variations of JNNTs by including only the node-level features and network-level features denoted as JNNTs-node and JNNTs-network, respectively. This allows us to evaluate the model performance when either of the components fails to be included in our predictive mechanism.
Similar to JNNTs, for competing methods involving tuning parameters, their optimal values are also determined under the validation set by $R^2$.
Eventually, we summarize the feature selection accuracy by sensitivity (Sens) and specificity (Spec) for both node- and network-level features, and assess predictive performance by $R^2$ under testing set.
Of note, since Horseshoe does not explicitly impose sparsity,
we implement a signal adaptive variable selector (SAVS) \citep{Ray2018} to determine the selected feature via soft-thresholding. 
We also try to impose sparsity by using an optimal cutoff selected by the validation set on the absolute value of the coefficients. Since both approaches lead to similar performance, we only reported the results for Horseshoe based on SAVS. The results for all the methods summarized across simulated settings are shown in Tables \ref{tbl:S12} -- \ref{tbl:S34}.
\begin{table}[htp]
    \centering
    \resizebox{\textwidth}{!}{
    \begin{tabular}{lcclccccc}
        \toprule
        Scenario  & $N$ & Noise & Method & $R^2$ & \multicolumn{2}{c}{Region} & \multicolumn{2}{c}{Network} \\
        \cline{6-9}
        \multicolumn{5}{c}{} & Sens (\%) & Spec (\%) & Sens (\%) & Spec (\%) \\
        \toprule
        \midrule
        \multirow{25}*{Scenario 1} & \multirow{12.5}*{50} &
        \multirow{6}*{Low}
        & JNNTs         & 0.950(0.082) & 0.819(0.121) & 0.936(0.090) & 0.991(0.047) & 0.992(0.016) \\
        & & & Lasso         & 0.506(0.190) & 0.326(0.155) & 0.881(0.119) & 0.483(0.155) & 0.886(0.054) \\
        & & & Horseshoeshoe         & 0.687(0.154) & 0.223(0.102) & 0.985(0.039) & 0.421(0.111) & 0.987(0.010) \\
        & & & JNNTs-node    & 0.041(0.033) & 0.534(0.286) & 0.614(0.320) &     --       &     --       \\
        & & & JNNTs-network & 0.704(0.150) &      --      &     --       & 0.658(0.148) & 0.980(0.032) \\
        & & & Tensor & 0.426(0.201) &     --       &     --      & 0.721(0.237) & 0.598(0.285) \\
        \cmidrule{3-9}
        & & \multirow{6}*{High}
        & JNNTs         & 0.629(0.172) & 0.463(0.210) & 0.845(0.156) & 0.633(0.210) & 0.965(0.064) \\
        & & & Lasso         & 0.347(0.146) & 0.250(0.147) & 0.915(0.104) & 0.381(0.162) & 0.919(0.057) \\
        & & & Horseshoe         & 0.406(0.144) & 0.122(0.061) & 0.975(0.047) & 0.280(0.091) & 0.985(0.011) \\
        & & & JNNTs-node    & 0.043(0.031) & 0.496(0.276) & 0.617(0.326) &      --      &      --      \\
        & & & JNNTs-network & 0.478(0.175) &     --       &     --       & 0.437(0.185) & 0.977(0.034) \\
        & & & Tensor & 0.336(0.165) &     --       &     --      & 0.694(0.203) & 0.602(0.241) \\
        \cmidrule{2-9}
        & \multirow{12.5}*{200} & \multirow{6}*{Low}
        & JNNTs         & 0.979(0.028) & 0.988(0.034) & 0.990(0.031) & 1.000(0.000) & 0.997(0.008) \\
        & & & Lasso         & 0.978(0.004) & 0.999(0.011) & 0.732(0.136) & 0.989(0.024) & 0.719(0.064) \\
        & & & Horseshoe         & 0.982(0.004) & 0.995(0.022) & 0.931(0.075) & 0.977(0.031) & 0.926(0.024) \\
        & & & JNNTs-node    & 0.064(0.032) & 0.446(0.212) & 0.832(0.169) &     --       &      --      \\
        & & & JNNTs-network & 0.804(0.141) &     --       &     --       & 0.966(0.107) & 0.991(0.017) \\
        & & & Tensor & 0.771(0.068) &     --       &     --      & 0.998(0.011) & 0.229(0.226) \\
        \cmidrule{3-9}
        & & \multirow{6}*{High}
        & JNNTs         & 0.874(0.023) & 0.799(0.099) & 0.958(0.064) & 1.000(0.000) & 0.995(0.013) \\
        & & & Lasso         & 0.834(0.026) & 0.844(0.118) & 0.763(0.147) & 0.873(0.066) & 0.774(0.075) \\
        & & & Horseshoe         & 0.848(0.026) & 0.779(0.105) & 0.928(0.078) & 0.828(0.068) & 0.919(0.020)  \\
        & & & JNNTs-node    & 0.061(0.034) & 0.491(0.200) & 0.786(0.190) &    --        &      --      \\
        & & & JNNTs-network & 0.742(0.075) &     --       &     --       & 0.919(0.131) & 0.987(0.019) \\
        & & & Tensor & 0.672(0.075) &     --       &     --      & 0.998(0.013) & 0.158(0.155) \\
        \midrule
        \midrule
        \multirow{25}*{Scenario 2} & \multirow{12.5}*{50} & \multirow{6}*{Low}
        & JNNTs          & 0.930(0.139) & 0.765(0.150) & 0.931(0.094) & 0.981(0.069) & 0.988(0.022) \\
        & & & Lasso          & 0.549(0.168) & 0.278(0.153) & 0.900(0.100) & 0.482(0.144) & 0.892(0.059) \\
        & & & Horseshoe          & 0.690(0.128) & 0.176(0.082) & 0.987(0.037) & 0.432(0.091) & 0.986(0.010) \\
        & & & JNNTs-node     & 0.041(0.036) & 0.577(0.324) & 0.520(0.373) &      --      &      --      \\
        & & & NJNNTs-network & 0.737(0.099) &    --        &      --      & 0.711(0.225) & 0.980(0.020) \\
        & & & Tensor & 0.444(0.231) &     --       &     --      & 0.753(0.218) & 0.579(0.236) \\
        \cmidrule{3-9}
        & & \multirow{6}*{High}
        & JNNTs         & 0.586(0.219) & 0.415(0.203) & 0.783(0.176) & 0.632(0.240) & 0.962(0.035) \\
        & & & Lasso         & 0.395(0.144) & 0.241(0.147) & 0.912(0.091) & 0.400(0.154) & 0.906(0.061) \\
        & & & Horseshoe         & 0.410(0.142) & 0.148(0.067) & 0.988(0.031) & 0.296(0.081) & 0.982(0.013) \\
        & & & JNNTs-node    & 0.038(0.029) & 0.542(0.316) & 0.592(0.362) &     --       &      --       \\
        & & & JNNTs-network & 0.521(0.167) &    --        &      --      & 0.535(0.277) & 0.967(0.045) \\
        & & & Tensor & 0.322(0.157) &     --       &     --      & 0.703(0.205) & 0.617(0.246) \\
        \cmidrule{2-9}
        & \multirow{12.5}*{200} & \multirow{6}*{Low}
        & JNNTs         & 0.976(0.047) & 0.987(0.036) & 0.985(0.034) & 1.000(0.000) & 0.998(0.006) \\
        & & & Lasso         & 0.978(0.004) & 0.993(0.027) & 0.724(0.128) & 0.991(0.020) & 0.723(0.057) \\
        & & & Horseshoe         & 0.983(0.003) & 0.988(0.035) & 0.933(0.075) & 0.978(0.032) & 0.930(0.023) \\
        & & & JNNTs-node    & 0.062(0.034) & 0.450(0.210) & 0.851(0.148) &     --       &     --       \\
        & & & JNNTs-network & 0.803(0.126) &    --        &    --        & 0.973(0.091) & 0.990(0.019) \\
        & & & Tensor &  0.764(0.059) &     --       &     --      & 1.000(0.000) & 0.194(0.209) \\
        \cmidrule{3-9}
        & & \multirow{6}*{High}
        & JNNTs         & 0.872(0.026) & 0.789(0.112) & 0.936(0.069) & 0.993(0.044) & 0.994(0.012) \\
        & & & Lasso         & 0.834(0.028) & 0.832(0.113) & 0.755(0.141) & 0.879(0.063) & 0.767(0.070) \\
        & & & Horseshoe         & 0.853(0.023) & 0.790(0.097) & 0.923(0.086) & 0.819(0.069) & 0.917(0.018) \\
        & & & JNNTs-node    & 0.067(0.035) & 0.463(0.212) & 0.803(0.145) &     --       &     --       \\
        & & & JNNTs-network & 0.735(0.099) &     --       &     --       & 0.916(0.128) & 0.989(0.018) \\
        & & & Tensor &  0.663(0.080) &     --       &     --      & 0.998(0.010) & 0.159(0.153) \\
        \midrule
        \bottomrule
    \end{tabular}}
    \caption{Simulation results for Scenario 1 and Scenario 2. Reported are the average values over 100 simulation runs, with the standard deviations in parentheses.}
    \label{tbl:S12}
\end{table}
\begin{table}[htp]
    \centering
    \resizebox{\textwidth}{!}{
    \begin{tabular}{lcclccccc}
        \toprule
        Scenario  & $N$ & Noise & Method & $R^2$ & \multicolumn{2}{c}{Region} & \multicolumn{2}{c}{Network} \\
        \cline{6-9}
        \multicolumn{5}{c}{} & Sens (\%) & Spec (\%) & Sens (\%) & Spec (\%) \\
        \toprule
        \midrule
        \multirow{25}*{Scenario 3} & \multirow{12.5}*{50} &
        \multirow{6}*{Low}
        & JNNTs             & 0.906(0.118) & 0.899(0.158) & 0.984(0.052) & 0.942(0.119) & 0.979(0.018) \\
        & & & Lasso         & 0.572(0.172) & 0.310(0.183) & 0.842(0.195) & 0.333(0.115) & 0.862(0.101) \\
        & & & Horseshoeshoe & 0.700(0.168) & 0.405(0.150) & 0.877(0.102) & 0.665(0.104) & 0.876(0.040) \\
        & & & JNNTs-node    & 0.048(0.040) & 0.557(0.319) & 0.542(0.360) &     --       &     --       \\
        & & & JNNTs-network & 0.681(0.168) &      --      &     --       & 0.666(0.191) & 0.976(0.030) \\
        & & & Tensor        & 0.439(0.212) &     --       &     --       & 0.762(0.195) & 0.573(0.252) \\
        \cmidrule{3-9}
        & & \multirow{6}*{High}
        & JNNTs         & 0.541(0.228) & 0.553(0.205) & 0.941(0.095) & 0.659(0.170) & 0.953(0.040) \\
        & & & Lasso         & 0.372(0.148) & 0.202(0.141) & 0.789(0.319) & 0.255(0.116) & 0.880(0.107) \\
        & & & Horseshoe         &  0.400(0.168) & 0.393(0.147) & 0.775(0.142) & 0.280(0.091) & 0.985(0.011) \\
        & & & JNNTs-node    & 0.041(0.037) & 0.529(0.319) & 0.579(0.356) &      --      &      --      \\
        & & & JNNTs-network & 0.411(0.183) &     --       &     --       & 0.429(0.210) & 0.979(0.028) \\
        & & & Tensor & 0.290(0.172) &     --       &     --       & 0.741(0.196) & 0.525(0.238) \\
        \cmidrule{2-9}
        & \multirow{12.5}*{200} & \multirow{6}*{Low}
        & JNNTs         & 0.956(0.084) & 0.997(0.018) & 0.998(0.015) & 0.988(0.042) & 0.988(0.008) \\
        & & & Lasso         & 0.979(0.004) & 0.991(0.030) & 0.733(0.152) & 0.616(0.098) & 0.704(0.060) \\
        & & & Horseshoe         & 0.982(0.004) & 0.995(0.022) & 0.931(0.075) & 0.977(0.031) & 0.926(0.024) \\
        & & & JNNTs-node    & 0.069(0.044) & 0.499(0.204) & 0.810(0.175) &     --       &      --      \\
        & & & JNNTs-network & 0.799(0.146) &     --       &     --       & 0.958(0.103) & 0.983(0.011) \\
        & & & Tensor & 0.770(0.067) &     --       &     --       & 0.998(0.014) & 0.229(0.214) \\
        \cmidrule{3-9}
        & & \multirow{6}*{High}
        & JNNTs         & 0.854(0.060) & 0.931(0.104) & 0.992(0.029) & 0.979(0.054) & 0.985(0.011) \\
        & & & Lasso         & 0.831(0.030) & 0.849(0.108) & 0.775(0.150) & 0.527(0.095) & 0.745(0.077) \\
        & & & Horseshoe         & 0.848(0.026) & 0.779(0.105) & 0.928(0.078) & 0.828(0.068) & 0.919(0.020)  \\
        & & & JNNTs-node    & 0.062(0.040) & 0.476(0.205) & 0.812(0.167) &    --        &      --      \\
        & & & JNNTs-network & 0.703(0.143) &     --       &     --       & 0.899(0.129) & 0.981(0.017) \\
        & & & Tensor & 0.671(0.075) &     --       &     --       & 0.995(0.019) & 0.163(0.171) \\
        \midrule
        \midrule
        \multirow{12.5}*{Scenario 4} & \multirow{12.5}*{1000} & \multirow{6}*{Low}
        & JNNTs         & 0.971(0.031) & 1.000(0.000) & 0.996(0.007) & 1.000(0.000) & 0.999(0.004) \\
        & & & Lasso         & 0.931(0.011) & 0.999(0.005) & 0.882(0.042) & 0.988(0.008) & 0.882(0.014) \\
        & & & Horseshoe         & 0.970(0.005) & 1.000(0.000) & 0.998(0.005) & 0.981(0.009) & 0.998(0.001) \\
        & & & JNNTs-node    & 0.022(0.014) & 0.192(0.127) & 0.999(0.002) &      --      &     --       \\
        & & & JNNTs-network & 0.919(0.014) &    --        &     --      & 1.000(0.001) & 1.000(0.000) \\
        & & & Tensor &  0.737(0.174) &     --       &     --      & 0.997(0.007) & 0.420(0.281) \\
        \cmidrule{3-9}
        & & \multirow{6}*{High}
        & JNNTs         & 0.797(0.042) & 0.835(0.102) & 0.984(0.015) & 1.000(0.000) & 0.999(0.001) \\
        & & & Lasso         & 0.581(0.047) & 0.682(0.112) & 0.923(0.030) & 0.801(0.033) & 0.923(0.011) \\
        & & & Horseshoe         & 0.657(0.050) & 0.568(0.104) & 0.984(0.014) & 0.739(0.035) & 0.985(0.002) \\
        & & & JNNTs-node    & 0.016(0.015) & 0.180(0.118) & 0.999(0.001) &      --      &     --       \\
        & & & JNNTs-network & 0.713(0.160) &     --       &     --       & 0.983(0.078) & 0.999(0.002) \\
        & & & Tensor &  0.596(0.167) &     --       &     --      & 0.997(0.010) & 0.195(0.089) \\
        \midrule
        \bottomrule
    \end{tabular}}
    \caption{Simulation results for Scenario 3 and Scenario 4. The layout is the same as in Table \ref{tbl:S12}.}
    \label{tbl:S34}
\end{table}

Based on the results, we conclude that JNNTs outperforms the other methods in out-of-sample prediction and selection for both types of features in all the simulated settings. This confirms the importance of jointly considering the two feature components and accommodating the topological structure of network-variate. Particularly, under the most challenging settings with  small sample sizes and large noise levels, JNNTs can still provide a satisfactory performance shown by the large $R^2$, sensitivities and specificities, which ensures its general use in real practice. When comparing Scenarios 1 and 2, we notice the performance of JNNTs is quite similar between the two scenarios. Despite assuming the coupling structure, JNNTs can uncover the effects equally well when the two components are independent. A further comparison between Scenarios 1 and 3 shows that JNNT maintains a satisfactory performance with mis-specified sub-network architectures. These indicate the robustness of our method with respect to the coupling effect and network figure configurations.
Under Setting 4 with a data scale close to our ABCD application, JNNTs remains a strong predictive power and high feature selection accuracy despite the noise level. Finally, among the competing methods, JNNTs-network achieves the best prediction and network feature selection results almost in all the settings. This further emphasizes the importance of accommodating network configuration for connectivity predictors. For JNNTs-node, despite a relatively robust performance of feature selection under small sample sizes, it suffers from poor predictive performance without a contribution of network components. Lasso and Horseshoe show a similar pattern under most of the settings with a satisfactory performance under large samples but a significant deterioration in predicting new samples and identifying the true signals with sample size decreased.
Tensor on the other hand suffers from the over-selection issue and worse predictive accuracy compared with JNNTs-network.

We also perform an extra simulation study to better mimic the neuroimaging study, denoted as Setting 2. In this setting, the observed brain regional and network measurements from the ABCD study is directly used as the design matrix to generate data, with $P=264$. To demonstrate the performance of JNNTs on different data size, we set $N=200$ and $N=500$ and randomly sampling without replace from the original ABCD study to construct $\X$ and $\Z$.
As to the phenotype,
we first set the signal patterns
where there are 30 nodes and two sub-networks associated with the phenotype, and then generate coefficient vectors as described in simulation Scenario 1--5. In this setting, we also consider two noise levels. The results are summarized in Table \ref{tbl:S_real}.
Due to relative large feature dimension and small sample size,
Tensor method is not applicable when $N=200$, and can only handle rank 1 fit with rank 1 tensor coefficient, i.e., $A = \balpha^{(1)} \otimes \balpha^{(1)}$, when $N=500$.
However, the underlying coefficient matrix has rank 2 and hence tensor method has much deteriorate performance.
All methods deteriorates when noise level increases and sample size decreases.
Due to the correlation between regional measurements, all methods have lower Sensitivity in regional selection comparing to Scenario 1 -- 4.
However, JNNTs outperforms all the other competing methods substantially, in terms of both prediction accuracy and Sensitivity, especially for the smaller datasets.
\begin{table}[htp]
  \centering
  \resizebox{\textwidth}{!}{
  \begin{tabular}{lcllcccc}
      \toprule
      $N$ & Noise & Method & $R^2$ & \multicolumn{2}{c}{Region} & \multicolumn{2}{c}{Network} \\
      \cline{4-8}
      \multicolumn{4}{c}{} & Sens (\%) & Spec (\%) & Sens (\%) & Spec (\%) \\
      \toprule
      \midrule
      \multirow{12.5}*{200} &
      \multirow{6}*{Low}
      & 
  JNNTs         & 0.799(0.061) & 0.542(0.093) & 0.990(0.007) & 0.856(0.082) & 0.998(0.003) \\
      & & 
  Lasso         & 0.229(0.081) & 0.160(0.035) & 0.966(0.008) & 0.157(0.023) & 0.998(0.001) \\
      & & 
  Horseshoe     & 0.113(0.182) & 0.110(0.054) & 0.919(0.029) & 0.197(0.033) & 0.986(0.002) \\
      & & 
  JNNTs-node    & 0.104(0.032) & 0.246(0.212) & 0.669(0.229) &     --       &      --      \\
      & & 
  JNNTs-network & 0.628(0.141) &     --       &     --       & 0.642(0.107) & 0.891(0.017) \\
      & & 
  Tensor & -- & -- & -- & -- & -- \\ 
      \cmidrule{2-8}
      & \multirow{6}*{High}
      & 
  JNNTs         & 0.647(0.095) & 0.425(0.067) & 0.923(0.083) & 0.802(0.091) & 0.920(0.012)) \\
      & & 
  Lasso         & 0.116(0.073) & 0.052(0.033) & 0.917(0.009) & 0.090(0.028) & 0.941(0.021) \\
      & & 
  Horseshoe     & 0.069(0.200) & 0.095(0.069) & 0.889(0.030) & 0.077(0.030) & 0.921(0.052)  \\
      & & 
  JNNTs-node    & 0.083(0.034) & 0.143(0.186) & 0.794(0.190) &    --        &      --      \\
      & & 
  JNNTs-network & 0.495(0.205) &     --       &     --       & 0.603(0.171) & 0.827(0.780) \\
      & & 
  Tensor & -- & -- & -- & -- & -- \\
  \midrule \midrule
      \multirow{12.5}*{500} & \multirow{6}*{Low}
      & 
  JNNTs         & 0.936(0.047) & 0.761(0.054) & 0.995(0.004) & 0.894(0.024) & 1.000(0.000) \\
      & & 
  Lasso         & 0.547(0.070) & 0.253(0.060) & 0.939(0.019) & 0.341(0.037) & 0.991(0.003) \\
      & & 
  Horseshoeshoe & 0.616(0.124) & 0.305(0.069) & 0.931(0.026) & 0.378(0.034) & 0.981(0.003) \\
      & & 
  JNNTs-node    & 0.196(0.051) & 0.381(0.104) & 0.799(0.137) &     --       &     --       \\
      & & 
  JNNTs-network & 0.704(0.150) &      --      &     --       & 0.658(0.148) & 0.980(0.032) \\
      & & 
   Tensor       & 0.102(0.201) &     --       &     --       & 0.421(0.203) & 0.398(0.241) \\
      \cmidrule{2-8}
      & \multirow{6}*{High}
      & 
  JNNTs         & 0.814(0.072) & 0.628(0.069) & 0.986(0.009) & 0.806(0.119) & 0.965(0.064) \\
      & & 
  Lasso         & 0.427(0.068) & 0.196(0.056) & 0.842(0.018) & 0.254(0.034) & 0.892(0.003)) \\
      & & 
  Horseshoe    &  0.506(0.122) & 0.283(0.067) & 0.901(0.018) & 0.301(0.046) & 0.928(0.003) \\
      & & 
  JNNTs-node    & 0.132(0.031) & 0.296(0.176) & 0.718(0.167) &      --      &      --      \\
      & & 
  JNNTs-network & 0.578(0.175) &     --       &     --       & 0.607(0.185) & 0.947(0.034) \\
      & & 
  Tensor        & 0.083(0.245) &     --       &     --      & 0.394(0.237) & 0.302(0.285) \\
      \midrule
      \bottomrule
  \end{tabular}
  }
  \caption{Simulation results for Setting 2 that mimics real data. The layout is the same as in Table \ref{tbl:S12}.}
  \label{tbl:S_real}
\end{table}

\section{Data application on the ABCD study} \label{sec::app}
We implement our JNNTs model to a recent large-scale children's brain developmental dataset to study the brain-to-behavior predictive mechanism and identify regional and network neuromarkers.
The ABCD 
study is an ongoing landmark study aiming to comprehensively characterize the biological and environmental factors for human brain development from childhood through young adulthood \citep{garavan2018recruiting}. This unprecedented cohort comprises over 10,000 early elementary school children from 21 recruitment sites across the United States, 
and there is a plan to follow up with them longitudinally for ten years to investigate their brain and mental development. In this work, we focus on the baseline data which contains extensive assessments of behavioral, psychosocial , and neuroimaging measurements. For the brain imaging data, we consider both resting-state and task-based fMRI released via Fast Track option as of April 2018. The resting-state has been playing an essential role in exhibiting intrinsic brain functional architecture, and serves as the most common resource for constructing brain functional connectivity \citep{lee2013resting}. For the task-based fMRI study, there is an emotional version of the n-back task where participants respond to whether the facial expression stimuli shown on the screen is the same as the one shown two trials earlier (2-back) or the same as  the one shown at the start of the block (0-back) \citep{cohen2016impact}. With the task engaging the neural correlates of working memory and emotional regulation processes, the task contrast between 2-back and 0-back provides important characterizations on an individual's working memory and executive processes. Therefore, it is of great interest to study how behavioral traits are impacted jointly by the resting-state brain networks and regional measurements of task contrast.

The obtained fMRI data have reflected the standard acquisition and preprocessing steps across different sites as described in the ABCD imaging acquisition and preprocessing protocols \citep{hagler2019image, casey2018adolescent}.  Briefly, the preprocessing steps for fMRI data included head motion correction, distortion correction, gradient warping correction, within-scan motion correction, and registration to structural images. Subsequently, all the fMRI images are realigned
and registered to the standard 2$mm$ MNI152 volumetric space, and then parcellated by a 264-node brain atlas \citep{Power2011} with mapping details shown in the Web Appendix C. The task contrast map across all the nodes is obtained directly using FSL's FEAT \citep{jenkinson2012fsl} to serve as the regional measurements. For the resting state, to construct functional connectivity, covariates of no interest including linear, quadratic, and cubic drifts, 24-motion parameters \citep{satterthwaite2013improved}, mean cerebral-spinal fluid signal, mean white matter signal, and overall global signal are first regressed out from the data. A temporal smoothing with a Gaussian filter (approximate cutoff frequency of 0.12 Hz) is performed afterward. Finally, to quantify functional connections, Pearson correlation coefficient between time courses from each pair of nodes is calculated, followed by a Fisher z-transformation to normalize the metric, which provides us with a $264\times 264$ connectivity matrix for each subject. For the phenotypic outcome, we consider the children's general intelligence, also known as the $\mathcal{g}$-factor, which plays a fundamental role in reflecting mental ability. For the non-imaging covariates, we include age, sex, race, parental  marital status, and parental highest education and income. In the end, a total of 1,894 subjects are included in our analyses with complete data.

We first assess the out-of-sample predictive performance. We randomly split the whole dataset ten times under a ratio 8:1:1 for a training set, a validation set, and a testing set, respectively. Following a similar operation procedure as in the simulation, we implement JNNTs in the training set with the number of sub-networks searched from $\{2, 3, 4, 5\}$.
Eventually, the average out-of-sample $R^2$ on the testing sets is 0.37 (0.05), which is highly promising compared with existing empirical results in the literature \citep{Dubois2018,zhao2022bayesian}. As a comparison, we also implement JNNTs-node and JNNTs-network following the same operation procedure and obtain 0.31 (0.04) and 0.23 (0.04), respectively. This result indicates the importance of jointly considering both node- and network-level neuromarkers and suggests that task contrast maps might contribute more significantly to informing general mental abilities.

To identify the signaling node and sub-network neuromarkers, we refit the model under the whole data set with $R=2$ which gives the highest out-of-sample $R^2$ among test sets. Eventually, we identify 30 node-level task contrast features and two sub-networks based on a 0.5 cutoff for the MPP, as shown in Figures \ref{fig:selection-ROI}  and \ref{fig:selection-Net}. A larger number of selections from the node-level features may support our previous results under each individual predictive component that activation contrast signals explain more phenotypic variations than connectivity traits.
We also attempt different kernels for $\bSigma_{\gamma}$ and $\bSigma_{\theta}$, and the corresponding results show quite similar selection for both region and network. See Table 2 
in Section D of the Supplementary Materials.
To further investigate the results, we first map each of the identified node signals to the canonical functional systems \citep{Power2011} as shown by different colors in Figure \ref{fig:selection-ROI}. Though the signals spread over different functional systems, one-third of them are concentrated in the Default Mode Network (DMN). This is consistent with existing neuroscience literature on the activity of the DMN during working memory tasks \citep{schultz2016higher} and the link of DMN regions with general intelligence \citep{anticevic2012role}. We further display the MPP of all the nodes in Figure \ref{fig:MPP_Reg}, which shows the high posterior likelihood for a majority of the identified features. For instance, the postcingulum and  postcentral on the left hemisphere receive the highest MPP, and both are crucial brain regions related to different memory-related functions, including pain and episodic memory retrieval, emotional salience, and spatial memory \citep{Maddock2003, Nielsen2005, Kozlovskiy2012}.
As for the two identified sub-networks, they share some overlapping network structures indicating that topological component might play a role in different neuronal communication processes.
The selected mid-cingulum on the right hemisphere shared by both sub-networks is also identified as a signaling task contrast location in the node-level selection.
Functionally, the mid-cingulum, as part of the memory retrieval module, is cytoarchitecturally located in the Brodmann area which occupies most of the posterior cingulate gyrus (PCC).
The PCC is widely known for its essential role in controlling social and emotional processing, and early neurodevelopment 
, which supports complex interactions among different neural networks \citep{Nielsen2005, Kozlovskiy2012}.
As a major hub of the DMN, the PCC communicates with various brain networks simultaneously with involvement in diverse functions \citep{Leech2014, Leech2012a, Pearson2011}. 
Finally, we also display the MPP for the network-level feature selection and highlight the selected edges under 99\%, 98\%, and 97\%
quantiles of the MPP in Figure \ref{fig:MPP_Net}.
The consistent pattern among different cutoffs indicates the existence of potential hubs, among which the mid-cingulum involves the largest number of connections.


\begin{figure}[!htp]
  \begin{subfigure}{0.95\textwidth}
      \centering
      \includegraphics[width=1\linewidth]{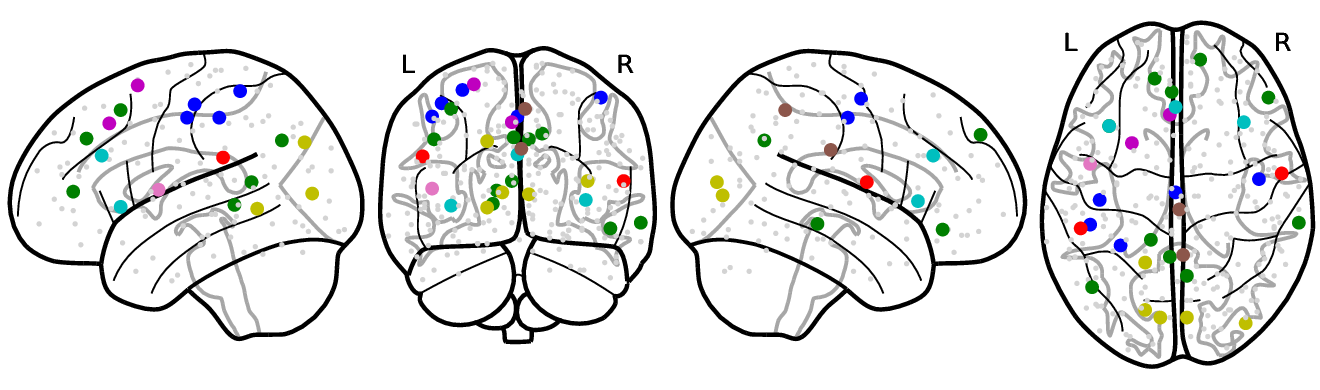}
  \end{subfigure}
  \begin{subfigure}{0.95\textwidth}
      \centering
      \includegraphics[width=.9\linewidth]{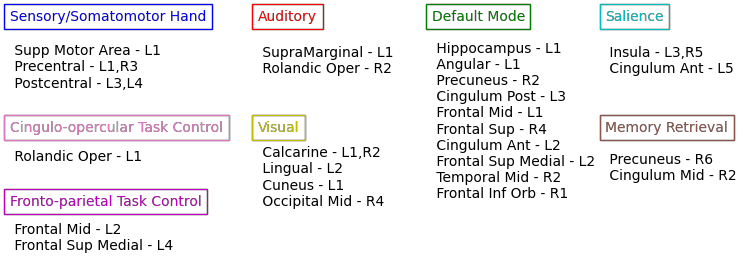}
  \end{subfigure}
  \caption{
      The identified nodes for brain regional features. From left to right are sagittal left hemisphere, coronal, sagittal right hemisphere, and axial views. The selected nodes are labeled by different colors according to which conventional functional system they belong to, and the remaining are showed by the small grey dots.}
  \label{fig:selection-ROI}
\end{figure}

\begin{figure}[!htp]
  \begin{subfigure}{0.74\textwidth}
  \centering
  \hspace{-1cm}
      \includegraphics[width=1\linewidth]{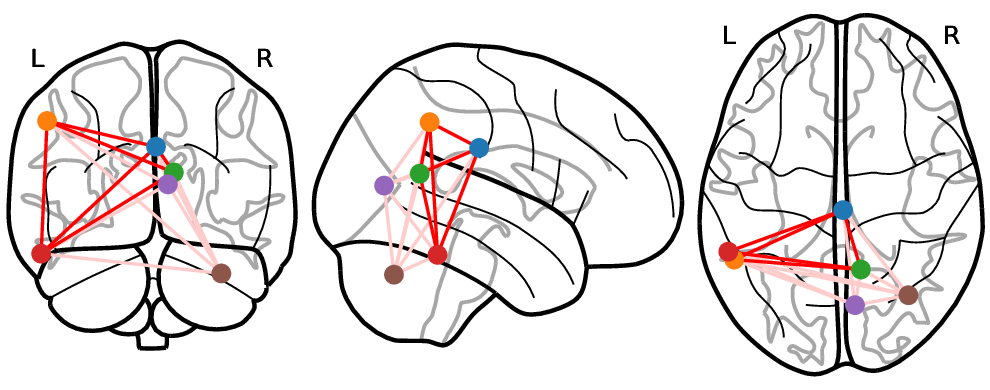}
      \caption{Sub-network 1.}
  \end{subfigure}
  \hfill
  \begin{subfigure}{0.74\textwidth}
  \centering
      \includegraphics[width=1\linewidth]{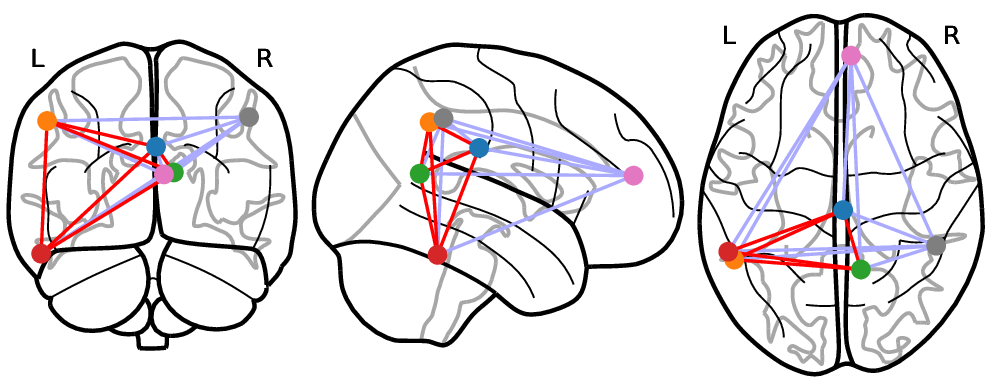}
      \caption{Sub-network 2.}
  \end{subfigure}
  \hfill
  \begin{subfigure}{0.24\textwidth}
      \centering
      \includegraphics[width=1\linewidth]{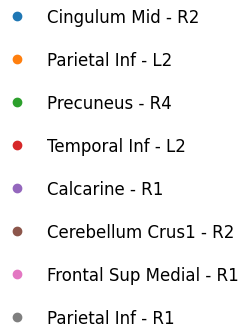}
  \end{subfigure}
  \caption{
  The identified two sub-networks from brain connectivity. The Shared edges are labeled in red, while unique edges are distinguished in pink and purple for sub-network 1 and sub-network 2, respectively.
  }
  \label{fig:selection-Net}
\end{figure}

\begin{figure}[!htp]
  \hspace{-2cm}
  \centering
  \includegraphics[width=1.1\linewidth]{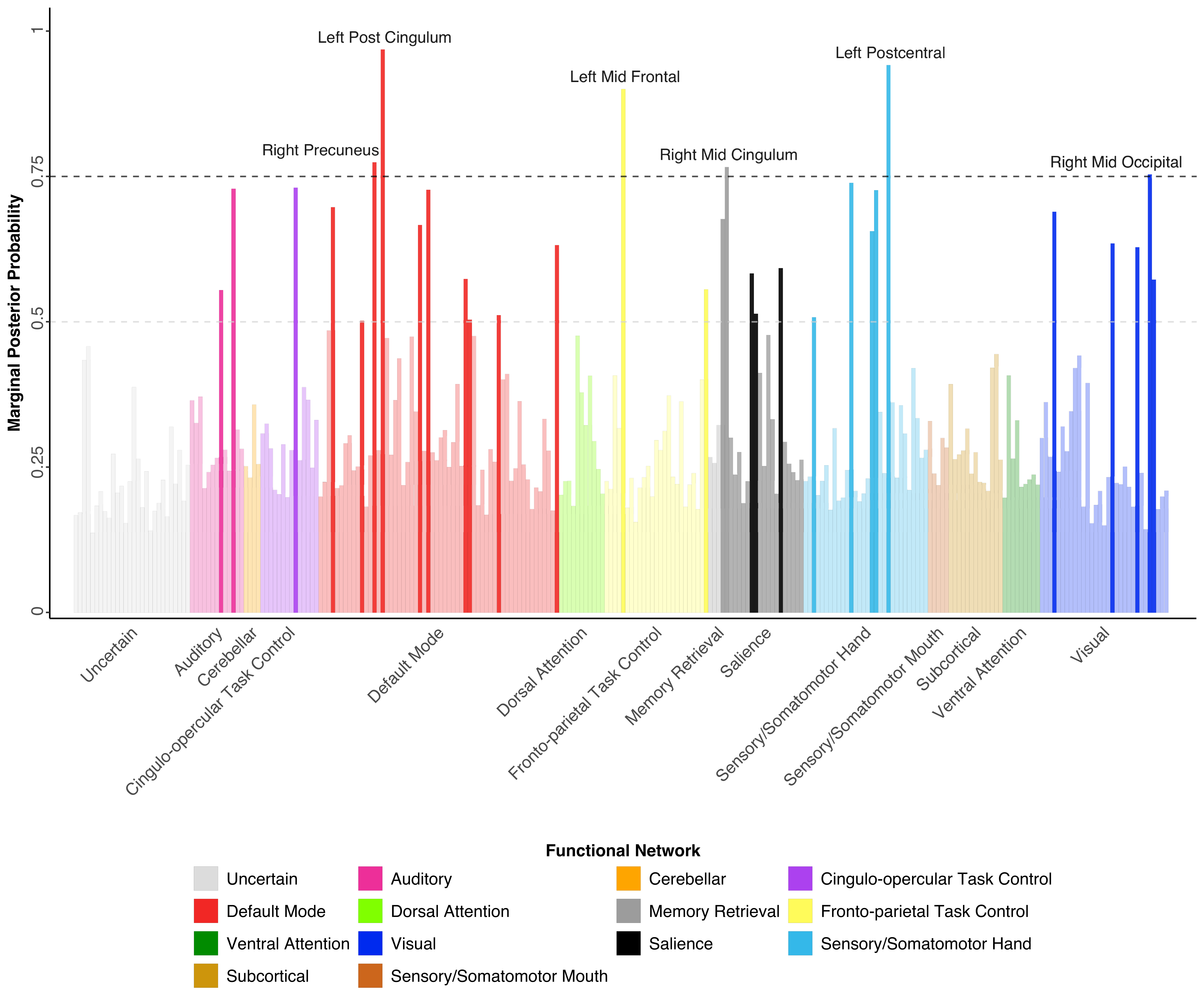}
  \caption{The MPP for node-level features. Nodes are grouped based on their corresponding canonical functional systems. The identified nodes, with MPP larger than 0.5, are displayed in opaque color. There are six nodes with MPP greater than 0.75, out of the total identified 30 nodes. }
  \label{fig:MPP_Reg}
\end{figure}

\begin{figure}[!htp]
  \begin{subfigure}{1\textwidth}
      \centering
      \caption{Identified sub-networks.}
      \label{fig:MMP_Subnet}
      \includegraphics[width=0.9\linewidth]{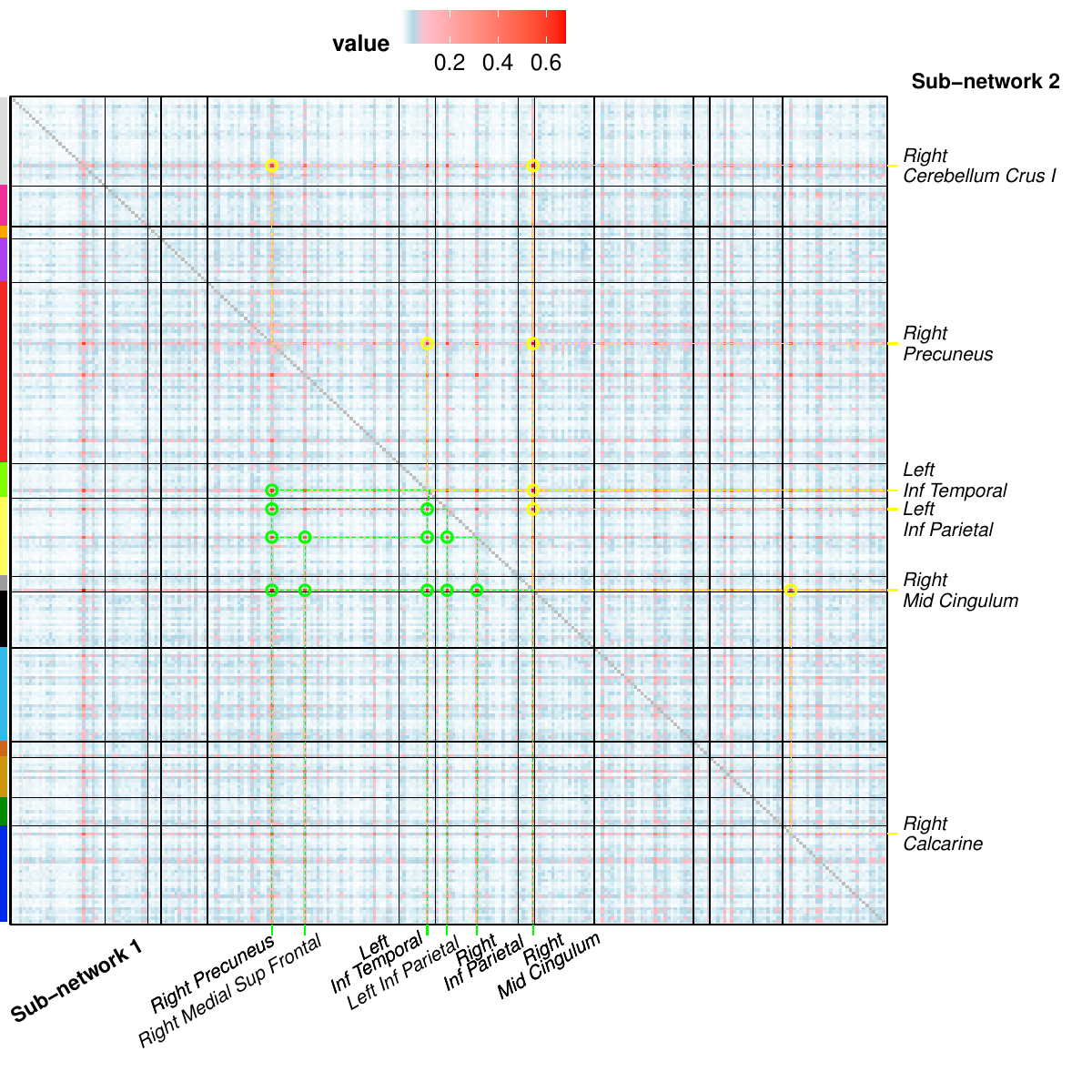}
  \end{subfigure}
  \hfill
  \begin{subfigure}{0.32\textwidth}
      \centering
      \caption{With 99\% percentile. }
      \includegraphics[width=1\linewidth]{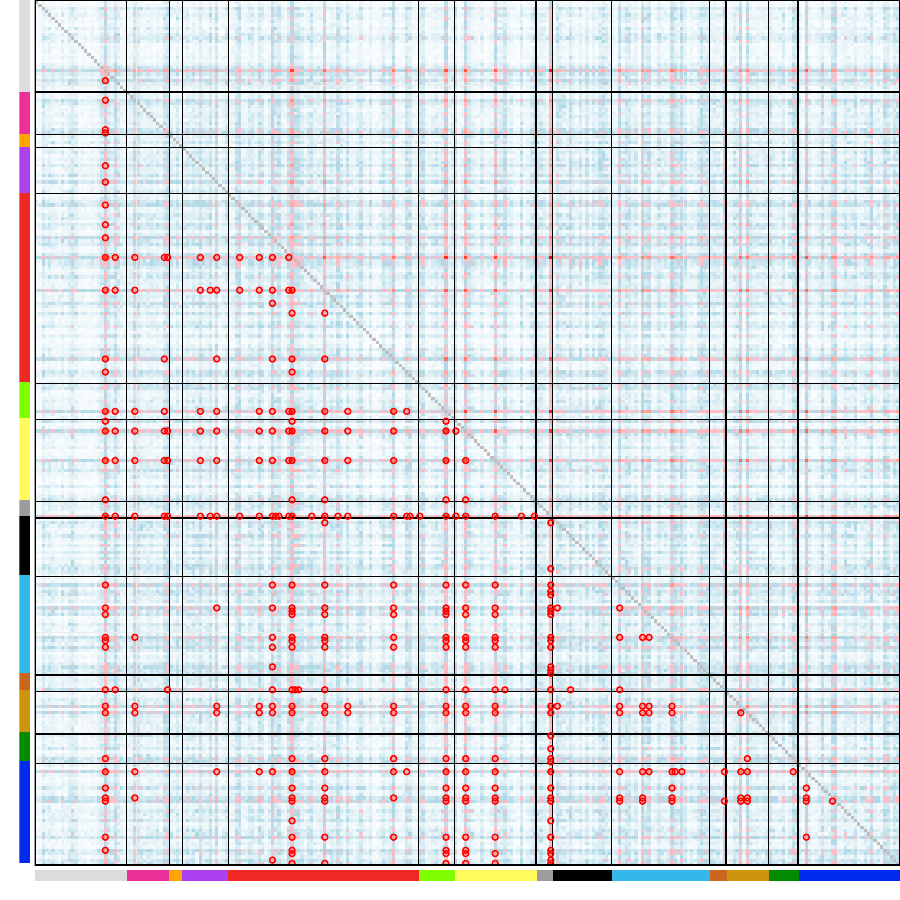}
      \label{fig: MMP_Net99}
  \end{subfigure}
  \begin{subfigure}{0.32\textwidth}
      \centering
      \caption{With 95\% percentile.}
      \includegraphics[width=1\linewidth]{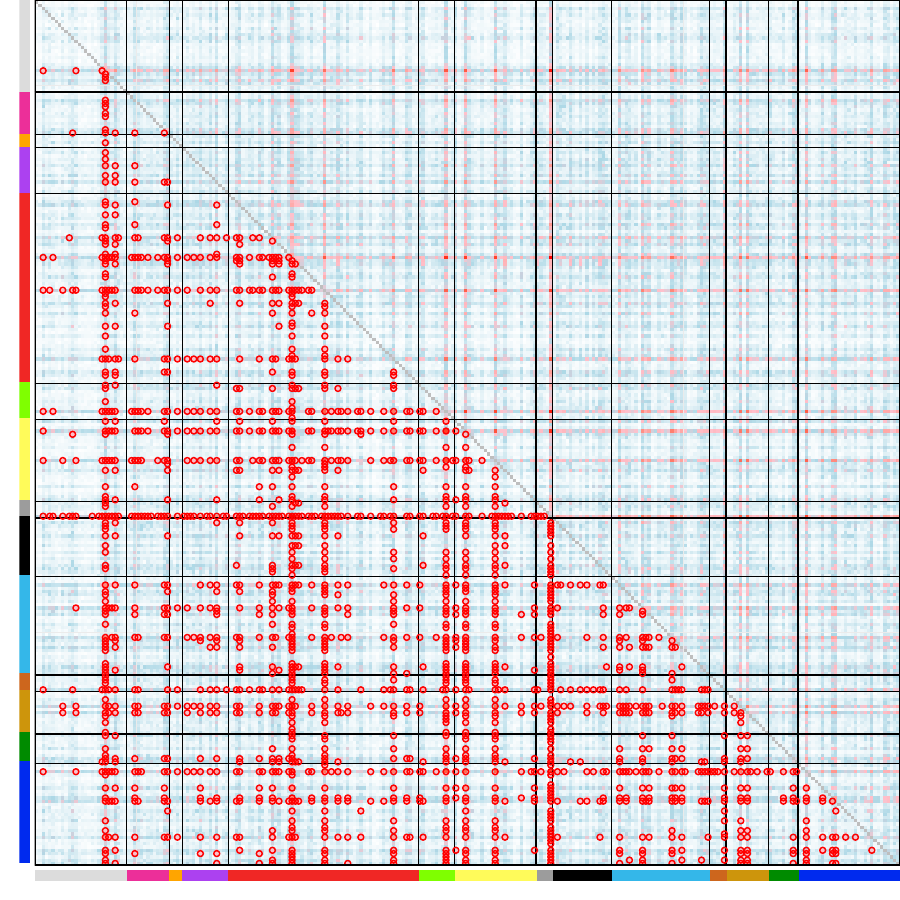}
      \label{fig: MMP_Net95}
  \end{subfigure}
  \begin{subfigure}{0.32\textwidth}
      \centering
      \caption{With 90\% percentile.}
      \includegraphics[width=1\linewidth]{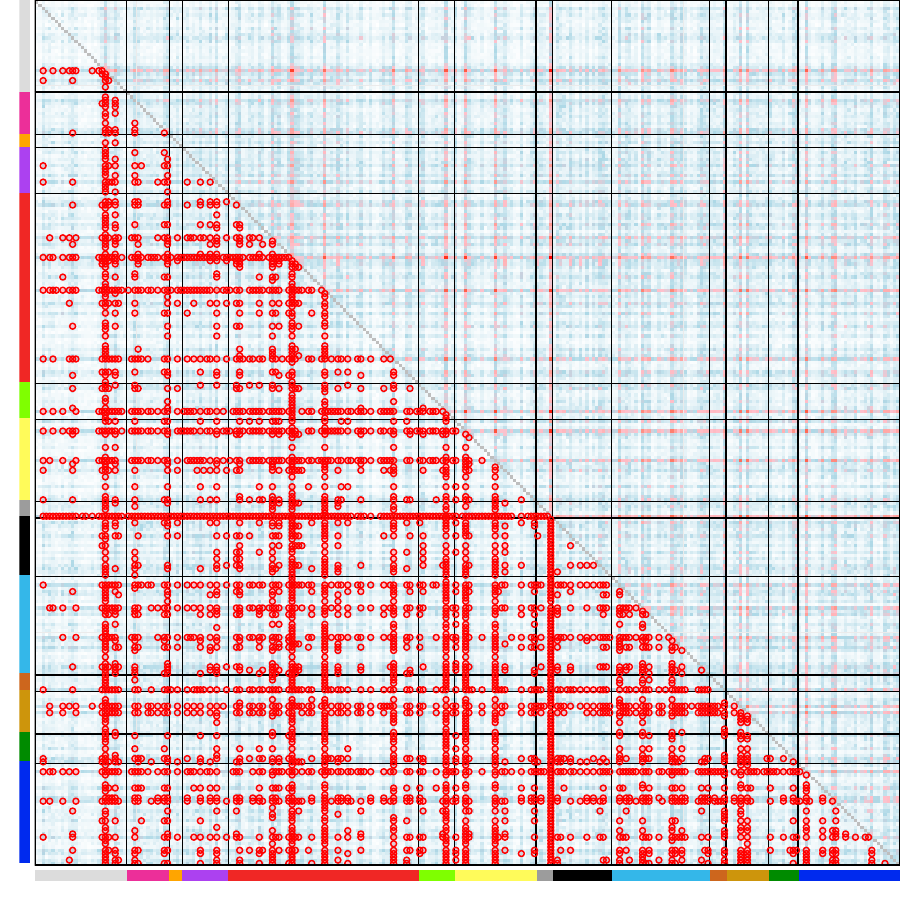}
      \label{fig: MMP_Net90}
  \end{subfigure}
  \caption{The MPP for network-level features.
      Nodes are grouped into different canonical functional systems with the same order and color scheme as in Figure \ref{fig:MPP_Reg}. 
      (a) The two \textcolor{red}{clique} sub-networks identified based on a 0.5 cutoff, displayed in upper- and lower-diagonal respectively. Edges in each sub-network with an MPP greater than 0.5 are circled in different colors, respectively.
      (b) -- (d) Edges with MPP greater than the cutoff values, 99\%, 98\% and 97\% quantiles of the MPP, are circled in red, respectively.
      }
  \label{fig:MPP_Net}
\end{figure}

\section{Discussion} \label{sec:dis}
In this paper, motivated by the urgent need in neuroscience studies to incorporate brain activities from different hierarchies, we address a gap in the literature by developing a unified Bayesian supervised learning framework that integrates vector-variate and matrix-variate predictors and accommodates the topological structures within the network features. To comprehensively
characterize the interplay that broadly exists between different predicting components, we propose a joint thresholded prior model to capture the coupling and grouping effect of signals, with consideration for their spatial contiguity across brain anatomy. Through extensive simulations, we demonstrate the superior performance of our method in prediction and feature selection. In our application to the landmark ABCD study, we provide a powerful predictive mechanism for children's general mental abilities by combining resting-state brain connectivity and memory task brain contrast maps, and identify signaling brain functional neuromarkers on the activation locations including postcingulum and postcentral areas, and intrinsic connectivity involving mid-cingulum.

There is always a trade-off in Bayesian models between the level of information that priors capture for desired knowledge and structure, and the model's robustness in handling general scenarios. In our current framework, despite the prior assumption on the coupling effect between attributes on the node and connectivity involving the node, the non-informative prior support for $\rho$ ensures the model performs equally well when signals are independent, as demonstrated by numerical studies. From another aspect, the hierarchical information in different brain locations may interact differently. For instance, we might see stronger correlations between activations and connections in the DMN under certain conditions given its dominant engagement in many cognitive activities \citep{smith2020default}. One possible extension of this work is to refine the global correlation parameter to make it spatially dependent. From a biological perspective, uncovering these correlations could provide valuable insights to better understand  the potential heterogeneity of neuronal functions.

Our current data application focuses on integrating different summarized levels of brain functional traits. As mentioned at the beginning of the paper, the model can be readily applicable to study node- and network-level structural traits that are anticipated to share a similar interactive mechanism and can also be accessed in neuroimaging studies including the ABCD. Meanwhile, along with the emerging interest to study the relationship between structural and functional imaging measurements, our model can also be slightly tailored to incorporate both modalities on top of different variates. This will allow a comprehensive analysis to characterize the underlying multi-way relationship among most of our commonly used neuromarkers. Furthermore, given that the ABCD study is prospective, it is also highly interesting to incorporate temporal effect in the analyses by investigating the dynamics of both node-level and network-level effects, as well as their interplay in future studies. Finally, with the availability of genotyping data, we are also working on potential extensions to incorporate genetic variants to build correspondence among genetic contributions, multi-variate brain activities and phenotypic profiles.


\section*{Acknowledgments} \label{acknowledgments}
\textit{Conflict of Interest}: None declared.

\section*{Funding} \label{funding}
This work was partially funded by awards RF1AG081413, R01EB034720 and  RF1AG068191 from the National Institutes of Health (NIH). The content is solely the responsibility of the authors and does not necessarily represent the official views of the NIH.

\clearpage
\appendix

{\noindent \Large\bf Supplementary Materials}

\section*{Section A. Proof of Theorem 1.}
Here we provide the proof for Theorem 1.
Let $[P] = \{1, 2, \cdots, P\}$. Define $g_r \subset [P], r=1, \cdots, R$ to be the support set of $\bfalpha^{(r)}$. That means, $g_r = \{j: \alpha^{(r)}_j \ne 0\}$. The resulting $\A=\sum_{r=1}^R\bfalpha^{(r)}\otimes\bfalpha^{(r)}$ thus has the following property:
\begin{equation}\label{eq:support-consistency}\tag{A.1}
a_{ij}\ne 0 ~\text{iff}~ \{i,j\} \subset g_r, \text{for some } r\in [R].
\end{equation}
The above \emph{support consistency} is generally true except on a negligible set of random events. So we treat all of our following discussions under this support-consistency scenario without loss of generality. In practice, condition \eqref{eq:support-consistency} is directly verifiable from data.

\begin{thm}
Suppose $\{g_r\}_{r=1}^R$ is support consistent with coefficient matrix $\bfA$ and each $g_r$ contains at least one unique variable that is not included in other sets. The clique set $\{g_r\}_{r=1}^R$ is the unique and minimum support consistent clique set of matrix $\bfA$
\end{thm}

\begin{proof}
For the ease of notations, define $g_{-r} = \cup_{r'\ne r}g_{r'}$. Suppose $\tilde{g}_k, k=1 \cdots, K$ is another support consistent clique set for $\A$. For each $r$, let $i_r$ be one variable in $g_r$ that does not exist in $g_{r'}, r'\ne r$. Suppose $i_r \in \tilde{g}_{k_r}$ for some $k_r \in [K]$.

First, because of the support consistency, for any $i \notin \cup_rg_r$, we have $a_{ij} = 0$ for all $j \in [P]$. Thus we know that either $i\notin \tilde{g}_{k}$ for all $k\in [K]$ or there exists a single set $\{i\}$ in $\{\tilde{g}_k\}$.  We will focus on the latter case by assuming that $\{\tilde{g}_k\}$ does not include new variables, because as can be seen clearly, including additional singleton sets will only increase the number of cliques.

By the support consistency of $\{g_r\}$ to $\A$, we know that $a_{i_r,j}=0$ for any $j \in g_{-r}\setminus g_{r}$. Therefore, it must be true that $\tilde{g}_{k_r} \cap (g_{-r}\setminus g_{r}) = \emptyset$, which is again, because of the support consistency. Thus we have
\begin{equation}\label{cond1}\tag{A.2}
\tilde{g}_{k_r}\subset g_r, r=1, \cdots, R.
\end{equation}

Meanwhile, define $\Omega_r = \{k: i_r \in \tilde{g}_k\}$. Since $\{\tilde{g}_k\}$ has to ensure the consistency with $a_{i_r,j} \ne 0, j \in g_r$, we have
\begin{equation}\label{cond2}\tag{A.3}
g_r \subset \cup_{k \in \Omega_r}\tilde{g}_k.
\end{equation}
Combining \eqref{cond1} and \eqref{cond2} leads to
$$ \cup_{k \in \Omega_r}\tilde{g}_k \subset g_r \subset  \cup_{k \in \Omega_r}\tilde{g}_k.$$
Thus we have for any $r$,
$$\cup_{k \in \Omega_r}\tilde{g}_k  = g_r.$$

Thus the only way to minimize the number of cliques $K$ is to set $K=R$ and $\tilde{g}_r = g_r$, up to a permutation of the indices.

\end{proof}

\section*{Section B. MCMC algorithm}
Recall that we adopt the squared exponential kernel structures for the covariance matrices and assume that $\sigma_{\gamma}=\sigma_{g}=\sigma$, i.e., $\bSigma_{\gamma}=\bSigma_{\theta}=\sigma\bO$, where we have each element of $\bO$ as $o_{ij}=\exp\{-\|\s_p - \s_j\|_2^2/2\}$ with $\s_p \in \mathbb{R}^{3}$ being the 3D coordinate for node $p  (p=1, \ldots, P)$.
Let $\bO = \bU\trans\diag\{\d\}\bU$, where $\bU = (\bu_1, \ldots, \bu_P) \in \mathbb{R}^{P \times P}$ and $\d = (d_1, \ldots, d_P)\trans \in \mathbb{R}^{P \times 1}$ are the matrix consisting of eigenvectors and vector of eigenvalues of $\bO$, respectively.
Assume $\bLambda = \rho\sigma_c \I_P=\rho(\sigma/\delta)\I_P$ as discussed in Section 3.1 in the main paper, with $\delta$ a pre-fixed large number.
The posterior inference can be easily adjusted for assigning different values for $\sigma_{\gamma}$,  $\sigma_{g}$ and $\sigma_c$.

We provide the detailed MCMC algorithm for the JNNTs model with identity link function based on Gibbs sampler and MH steps.
Let $\y = (y_{1}, \ldots, y_{N})\trans \in \mathbb{R}^{N \times 1}$, $\W=(\w_{1}, \ldots, \w_{N})\trans \in \mathbb{R}^{N \times Q}$,
$\X=(\x_1, \ldots, \x_N)\trans \in \mathbb{R}^{N \times P}$, and $\Z=(\Z_1, \ldots, \Z_N)\trans \in \mathbb{R}^{NP \times P}$. 
Also denote
\begin{align*}
    \bDelta_{\lambda}(\bgamma) = &
    \diag\Big\{\mathbb{I}(|\gamma_1|>\lambda), \ldots, \mathbb{I}(|\gamma_P|>\lambda)\Big\}, \\
    \bDelta_{\lambda}(\btheta, \btheta^{(r)}) = &
    \diag\Big\{\mathbb{I}(|\theta_1|>\lambda) \mathbb{I}(|\theta^{(r)}_1|>\lambda), \ldots, \mathbb{I}(|\theta_P|>\lambda) \mathbb{I}(|\theta^{(r)}_P|>\lambda)\Big\}, \\
    \bDelta_{\rho} = &
    \diag\Big\{1/(d_1^2 - \rho^2/\delta^2), \ldots, 1/(d_P^2 - \rho^2/\delta^2)\Big\}. 
\end{align*}
With random initials, we iteratively update parameters according to the following steps.  
\begin{itemize}
    \item[] Sampling scheme for $\bleta$. Draw
    \begin{align*}
        [\bleta \mid \mathcal{D}, \lambda, \bgamma, \wtbbeta, \btheta, \{\btheta^{(r)}\}_{r=1}^R, \{\wtbalpha^{(r)}\}_{r=1}^R, \sigma_{\epsilon}] \sim \mbox{N}(\bmu_{\eta}, \bSigma_{\eta}),
    \end{align*}
    with $\bSigma_{\eta} = \Big(\W\trans\W/\sigma_{\epsilon} + (1/\sigma_{\eta})\bI_{Q}\Big)^{-1}$,
    and $\bmu_{\eta} = (1/\sigma_{\epsilon}) \bSigma_{\eta}\W\trans\Big(\y - \X\bDelta_{\lambda}(\bgamma)\wtbbeta - \sum_{r=1}^R \big(\I_n \otimes \big( \bDelta_{\lambda}(\btheta, \btheta^{(r)})\wtbalpha^{(r)}\big)\trans \big)\Z \bDelta_{\lambda}(\btheta, \btheta^{(r)}) \wtbalpha^{(r)} \Big)$.
    \item[] Sampling scheme for $\wtbbeta$. Draw
        \begin{align*}
            [\wtbbeta \mid \mathcal{D}, \lambda, \bleta, \bgamma, \btheta, \{\btheta^{(r)}\}_{r=1}^R, \{\wtbalpha^{(r)}\}_{r=1}^R, \sigma^2_{\epsilon}, \sigma^2_{\beta}] \sim \mbox{N}(\bmu_{\beta}, \bSigma_{\beta}),
        \end{align*}
        with $\bSigma_{\beta} = \Big(\bDelta_{\lambda}(\bgamma) \X\trans \X \bDelta_{\lambda}(\bgamma)/\sigma_{\epsilon} + (1/\sigma_{\beta})\bI_{P}\Big)^{-1}$,
    and $\bmu_{\beta} = \bSigma_{\beta} \Big(\bDelta_{\lambda}(\bgamma) \X\trans
        \big(\y - \W\bleta - \sum_{r=1}^R \big(\I_n \otimes \big( \bDelta_{\lambda}(\btheta, \btheta^{(r)})\wtbalpha^{(r)}\big)\trans \big)\Z \bDelta_{\lambda}(\btheta, \btheta^{(r)}) \wtbalpha^{(r)} \big)/\sigma_{\epsilon} +
        \bgamma/\sigma_{\beta} \Big)$.
        \item[] Sampling scheme for ${\wtbalpha}^{(r)}$.
        Define
        \begin{align} \label{eq:e1}
        \xi_{i,r,p} =
        & y_i - \bleta\trans\w_i - \wtbbeta\trans\bDelta_{\lambda}(\bgamma) \x_i -\sum_{k \neq r}^R \big(\bDelta_{\lambda}(\btheta, \btheta^{(k)})\wtbalpha^{(k)}\big)\trans \Z_i \bDelta_{\lambda}(\btheta, \btheta^{(k)}) \wtbalpha^{(k)} \nn \\ \nonumber
        & - \sum_{j,l \neq p}^P \mathbb{I}(|\theta_j|>\lambda) \mathbb{I}(|\theta^{(r)}_j|>\lambda) \mathbb{I}(|\theta_l|>\lambda) \mathbb{I}(|\theta^{(r)}_l|>\lambda) z_{ijl} \talpha^{(r)}_{j}\talpha^{(r)}_{l}. \tag{B.1}
        \end{align}
        For $r=1, \ldots, R$ and $p=1, \ldots, P$, if $\mathbb{I}(|\theta_p|>\lambda) = 1$ and $\mathbb{I}(|\theta^{(r)}_p|>\lambda) = 1$, draw
        \begin{align*}
            [ \talpha_{p}^{(r)} \mid \mathcal{D}, \lambda, \bleta, \bgamma, \wtbbeta, \btheta, \{\btheta^{(r)}\}_{r=1}^R, \{\wtbalpha^{(k)}\}_{k \neq r}^P, \{\talpha_j^{(r)}\}_{j \neq p}^P, \sigma_{\epsilon}, \sigma_{\alpha}] \sim \mbox{N}(\mu_{\alpha_p^r}, \sigma_{\alpha_p^r}),
        \end{align*}
        with $\sigma_{\alpha_p^r} = \Big(1/\sigma_{\alpha} + \big(4/\sigma_{\epsilon}\big) \sum_{i}^{N}
        \big(\sum_{j \neq p}^P  \mathbb{I}(|\theta_j|>\lambda) \mathbb{I}(|\theta^{(r)}_j|>\lambda) z_{ipj}\talpha^{(r)}_{j}\big)^2
        \Big)^{-1}$,
        and $\mu_{\alpha_p^r} = \sigma_{\alpha_p^r} \Big(
            \theta_p/\sigma_{\alpha} + (2/\sigma_{\epsilon})\sum_{i=1}^{N} \xi_{i, r, p} \big(\sum_{j \neq p}^P  \mathbb{I}(|\theta_j|>\lambda) \mathbb{I}(|\theta^{(r)}_j|>\lambda) z_{ipj}\talpha^{(r)}_{j}\big)
        \Big)$;
        Otherwise draw from $\mbox{N}(\theta_p, \sigma_{\alpha})$.
        \item[] Sampling scheme for ${\btheta}^{(r)}$.
        For $r=1, \ldots, R$, $p =1, \ldots, P$, if $|\theta_p|>\lambda$, draw
        \begin{align*}
            [\theta^{(r)}_{p} \mid & \mathcal{D}, \lambda, \bleta, \bgamma, \wtbbeta, \btheta, \{\btheta^{(k)}\}_{k \neq r}, \{\theta^{(r)}_j\}_{j \neq p}, \{\wtbalpha^{(r)}\}_{r=1}^R, \sigma_{\epsilon}, \sigma_{\theta}] \\
            & \sim \psi^{\theta}_{-1}\mbox{TN}_{(-\infty, -\lambda)}(\theta_p, \sigma_{\theta}) + \psi^{\theta}_{0}\mbox{TN}_{(-\lambda, \lambda)}(\theta_p, \sigma_{\theta}) + \psi^{\theta}_{1}\mbox{TN}_{(\lambda, +\infty)}(\theta_p, \sigma_{\theta}),
        \end{align*}
        where $\mbox{TN}_{(a,b)}(\mu, \tau)$ is the truncated normal distribution with mean $\mu$, variance $\tau$ and support $[a, b]$,
        and $\psi^{\theta}_{k} = c^{\theta}_{k}/\sum_{k=-1,0,1}c^{\theta}_{k}$ with
        \begin{align*}
            c^{\theta}_{-1} & = \exp
            \Big\{-\frac{1}{2\sigma_{\epsilon}}
            \sum_{i=1}^{N} \big(
                \xi_{i,r,p} - 2\talpha^{(r)}_{p}\sum_{j \neq p}^P  \mathbb{I}(|\theta_j|>\lambda) \mathbb{I}(|\theta^{(r)}_j|>\lambda) z_{ipj}\talpha^{(r)}_{j} \big)^2
            \Big\}
            F\Big(-\frac{\lambda + \theta_p}{\sqrt{\sigma_{\theta}}}\Big), \\
            c^{\theta}_{0}  & = \exp
            \Big\{-\frac{1}{2\sigma_{\epsilon}}
                \sum_{i=1}^N\xi_{i,r,p}^2
            \Big\}
            \Big(F\big(\frac{\lambda - \theta_p}{\sqrt{\sigma_{\theta}}}\big) - F\big(-\frac{\lambda + \theta_p}{\sqrt{\sigma_{\theta}}} \big)\Big),\\
            c^{\theta}_{1}  & = \exp
            \Big\{-\frac{1}{2\sigma_{\epsilon}}
            \sum_{i=1}^{N} \big(
                \xi_{i,r,p} - 2\talpha^{(r)}_{p} \sum_{j \neq p}^P  \mathbb{I}(|\theta_j|>\lambda) \mathbb{I}(|\theta^{(r)}_j|>\lambda) z_{ipj}\talpha^{(r)}_{j} \big)^2
            \Big\}
            F\Big(-\frac{\lambda - \theta_p}{\sqrt{\sigma_{\theta}}}\Big), 
        \end{align*}
        where $F(x)$ is the cumulative distribution function (CDF) of the standard normal distribution, and $\xi_{i,r,p}$ is the same as in \eqref{eq:e1}; Otherwise draw $\theta_p \sim \mbox{N}(\theta_p, \sigma_{\theta})$.
        \item[] Sampling scheme for $\bgamma$. For $p = 1, \ldots, P$, draw
        \begin{align*}
            [\gamma_{p} \mid & \mathcal{D}, \lambda, \bleta, \{\gamma_j\}_{j \neq p}^P, \wtbbeta, \btheta, \{\btheta^{(r)}\}_{r=1}^R, \{\wtbalpha^{(r)}\}_{r=1}^R, \sigma_{\epsilon}, \sigma_{\beta}, \sigma, \rho ] \\
            & \sim \psi^{\gamma}_{-1} \mbox{TN}_{(-\infty, -\lambda)}(\mu_{\gamma_p}, \sigma_{\gamma_p}) + \psi^{\gamma}_{0} \mbox{TN}_{(-\lambda, \lambda)}(\mu_{\gamma_p}, \sigma_{\gamma_p}) + \psi^{\gamma}_{1} \mbox{TN}_{(\lambda, +\infty)}(\mu_{\gamma_p}, \sigma_{\gamma_p}),
        \end{align*}
        where $\sigma_{\gamma_p}=\Big(1/\sigma_{\beta} + (1/\sigma) \bu_{p}\trans \bDelta_{\rho}\diag\{\d\}\bu_{p}\Big)^{-1}$,
        $\mu_{\gamma_p}=\sigma_{\gamma_p} \Big(\tilde{\beta}_{p}/\sigma_{\beta} + (1/\sigma)\bu_{p}\trans \bDelta_{\rho} \big(\rho \bU\btheta - \diag\{\d\} \sum_{j \neq p}\gamma_j\bu_{j}\big) \Big)$,
        and $\psi^{\gamma}_{k} = c^{\gamma}_{k}/\sum_{k=-1,0,1}c^{\gamma}_{k}$ with
        \begin{align*}
            c_{-1}^{\gamma} & = \exp \Big\{-\frac{1}{2\sigma_{\epsilon}}
            \sum_{i=1}^{N}\big(\xi_{i, p}^{\gamma} - x_{ip}\tilde{\beta}_{p} \big)^2 \Big\}
            F\Big(-\frac{\lambda + \mu_{\gamma_p}}{\sqrt{\sigma_{\gamma_p}}}\Big), \\
            c_{0}^{\gamma}  & = \exp \Big\{-\frac{1}{2\sigma_{\epsilon}}\sum_{i=1}^N\big(\xi_{i, p}^{\gamma}\big)^2\Big\}
            \Big(F \big(\frac{\lambda - \mu_{\gamma_p}}{\sqrt{\sigma_{\gamma_p}}}\big) - F\big(-\frac{\lambda + \mu_{\gamma_p}}{\sqrt{\sigma_{\gamma_p}}}\big)\Big),\\
            c_{1}^{\gamma}  & = \exp\Big\{-\frac{1}{2\sigma_{\epsilon}}
            \sum_{i=1}^{N}\big(\xi_{i, p}^{\gamma} - x_{ip}\tilde{\beta}_{p} \big)^2 \Big\}
            F\Big(-\frac{\lambda - \mu_{\gamma_p}}{\sqrt{\sigma_{\gamma_p}}}\Big),
        \end{align*}
        where $\xi_{i, p}^{\gamma} = y_i - \bleta\trans\w_i - \sum_{j \neq p}^P x_{ij}\tilde{\beta}_{j}\mathbb{I}(|\gamma_j|>\lambda) - \sum_{r=1}^R (\wtbalpha^{(r)})\trans \bDelta_{\lambda}(\btheta, \btheta^{(r)}) \Z_i \bDelta_{\lambda}(\btheta, \btheta^{(r)}) \wtbalpha^{(r)}$.

        \item[] Sampling scheme for $\btheta$. For $p=1, \ldots, P$, draw
        \begin{align*}
            [\theta_{p} \mid & \mathcal{D}, \lambda, \bleta, \bgamma, \wtbbeta, \{\theta_{j}\}_{j \neq p}, \{\btheta^{(r)}\}_{r=1}^R, \{\wtbalpha^{(r)}\}_{r=1}^R, \sigma_{\epsilon}, \sigma_{\alpha}, \sigma_{\theta}, \sigma, \rho] \\
            & \sim \psi^{g}_{-1} \mbox{TN}_{(-\infty, -\lambda)}(\mu_{\theta_p}, \sigma_{\theta_p}) +\psi^{g}_{0}\mbox{TN}_{(-\lambda, \lambda)}(\mu_{\theta_p}, \sigma_{\theta_p}) + \psi^{g}_{-1}\mbox{TN}_{(\lambda, +\infty)}(\mu_{\theta_p}, \sigma_{\theta_p}),
        \end{align*}
        where $\sigma_{\theta_p}=\Big(R/\sigma_{\alpha} + R/\sigma_{\theta} + (1/\sigma)\bu_{p}\trans \bDelta_{\rho}\diag\{\d\}\bu_{p}\Big)^{-1}$,
        $\mu_{\theta_p}=\sigma_{\theta_p} \Big((1/\sigma_{\alpha})\sum_{r=1}^R\talpha^{(r)}_p + (1/\sigma_{\theta})\sum_{r=1}^R \ttheta^{(r)}_p + (1/\sigma)\bu_{p}\trans \bDelta_{\rho}\big(\rho \bU\bgamma - \diag\{\d\} \sum_{j\neq p}^P \theta_j\bu_{j} \big)\Big)$,
        and $\psi^{g}_{k} = c^{g}_{k}/\sum_{k=-1,0,1}c^{g}_{k}$ with
        \begin{align*}
            c^{g}_{-1} & = \exp
            \Big\{-\frac{1}{2\sigma_{\epsilon}}\sum_{i=1}^{N}
                \Big( \xi^{g}_{i, p} - 2\sum_{r=1}^{R} \mathbb{I}(|\theta^{(r)}_p|>\lambda)\talpha^{(r)}_{p} \big(\sum_{j \neq p}^P  \mathbb{I}(|\theta_j|>\lambda) \mathbb{I}(|\theta^{(r)}_j|>\lambda) z_{ipj}\talpha^{(r)}_{j} \big) \Big)^2
            \Big\}
            \\ & \qquad \quad \times
            F\Big(-\frac{\lambda + \mu_{\theta_p}}{\sqrt{\sigma_{\theta_p}}}\Big), \\
            c^{g}_{0}  & = \exp \Big\{-\frac{1}{2\sigma_{\epsilon}}\sum_{i=1}^{N}\big(\xi^{g}_{i, p}\big)^2\Big\}
            \Big(F\big(\frac{\lambda - \mu_{\theta_p}}{\sqrt{\sigma_{\theta_p}}}\big) - F\big(-\frac{\lambda + \mu_{\theta_p}}{\sqrt{\sigma_{\theta_p}}}\big)\Big),\\
            c^{g}_{1}  & = \exp
            \Big\{-\frac{1}{2\sigma_{\epsilon}}\sum_{i=1}^{N}
                \Big(\xi^{g}_{i, p} - 2\sum_{r=1}^{R} \mathbb{I}(|\theta^{(r)}_p|>\lambda)\talpha^{(r)}_{p} \big(\sum_{j \neq p}^P  \mathbb{I}(|\theta_j|>\lambda) \mathbb{I}(|\theta^{(r)}_j|>\lambda) z_{ipj}\talpha^{(r)}_{j} \big) \Big)^2
            \Big\}
            \\ & \qquad \quad \times
            F\Big(-\frac{\lambda - \mu_{\theta_p}}{\sqrt{\sigma_{\theta_p}}}\Big),
        \end{align*}
        where
        \begin{align*}
        \xi^{g}_{i, p} = & y_i - \bleta\trans\w_i \\
        & - \wtbbeta \trans \bDelta_{\lambda}(\bgamma) \x_i - \sum_{r=1}^R \sum_{j, \ell \neq p}^P \mathbb{I}(|\theta_{j}|>\lambda) \mathbb{I}(|\theta^{(r)}_{j}|>\lambda) \mathbb{I}(|\theta_{\ell}|>\lambda) \mathbb{I}(|\theta^{(r)}_{\ell}|>\lambda) \talpha^{(r)}_{j}\talpha^{(r)}_{\ell}z_{ij\ell}.
        \end{align*}
        \item [] Sampling scheme for $\sigma_{\beta}$. Draw
        \begin{align*}
            [\sigma_{\beta} \mid \wtbbeta, \bgamma] \sim \mbox{IG}\Big(a_{\beta} + P/2, b_{\beta} + (1/2)\|\wtbbeta - \bgamma\|_2^2\Big).
        \end{align*}
    \item [] Sampling scheme for $\sigma_{\alpha}$. Draw
        \begin{align*}
            [\sigma_{\alpha} \mid \{\wtbalpha^{(r)}\}_{r=1}^R, \btheta] \sim \mbox{IG}\Big(a_{\alpha} + PR/2, b_{\alpha} + (1/2)\sum_{r=1}^R\|\wtbalpha^{(r)} -\btheta\|_2^2\Big).
        \end{align*}
    \item [] Sampling scheme for $\sigma_{\theta}$. Draw
        \begin{align*}
            [\sigma_{\theta} \mid \{\btheta^{(r)}\}_{r=1}^R, \btheta] \sim \mbox{IG}\Big(a_{\theta} + PR/2, b_{\theta} + (1/2)\sum_{r=1}^R\|\btheta^{(r)}-\btheta\|_2^2\Big).
        \end{align*}
    \item [] Sampling scheme for $\sigma_{\epsilon}$. Draw
        \begin{align*}
            [ \sigma_{\epsilon} \mid \mathcal{D}, \lambda, \bleta, \bgamma, \wtbbeta, \btheta, \{\btheta^{(r)}\}_{r=1}^R, \{\wtbalpha^{(r)}\}_{r=1}^R] \sim
            \mbox{IG}\Big(a_{\epsilon} + N/2, b_{\epsilon} +(1/2) \bxi\trans \bxi \Big),
        \end{align*}
        where $\bxi = \y - \sum_{r=1}^R \big(\I_n \otimes \big( \bDelta_{\lambda}(\btheta, \btheta^{(r)})\wtbalpha^{(r)}\big)\trans \big) \Z \bDelta_{\lambda}(\btheta, \btheta^{(r)}) \wtbalpha^{(r)} -\W \bleta - \X\bDelta_{\lambda}(\bgamma)\wtbbeta$.
        \item[] Sampling scheme for $\sigma$. Draw
        \begin{align*}
            [ \sigma \mid \bgamma, \btheta, \rho] \sim  \mbox{IG} \Big(a_{\sigma} + P, b_{\sigma} + (1/2) \xi^{\rho} \Big),
        \end{align*}
        where $\xi^{\rho} = \big(\bU\bgamma\big)\trans \bDelta_{\rho}\diag\{\d\} \bU\bgamma + \big(\bU\btheta\big)\trans \bDelta_{\rho}\diag\{\d\} \bU\btheta - 2\rho\big(\bU\btheta\big)\trans \bDelta_{\rho}\bU\bgamma$.
        \item[] Sampling scheme for $\rho$. Since it's hard to get a close form for the conditional posterior distribution of $\rho$, we adopt Metropolis-Hastings algorithm (MH) \citep{MH1953} to update it. 
        Let $f_{\rho}(\rho)= \exp \Big\{ -\xi^{\rho} /(2\sigma) \Big\}\prod_{p=1}^P \Big(d_{p}^2 -\rho^2/\delta^2\Big)^{-1/2}$
        and we then have
        \begin{align*}
            \bpi(\rho \mid \bgamma, \btheta, \sigma)
            \propto f_{\rho}(\rho) \mbox{U}(-1, 1).
        \end{align*}
        Generate candidate $\rho^{'}$ based on random walk proposal.
        If $|\rho^{'}| > 1$, we keep current value of $\rho$; Otherwise calculate $R_{\rho} = f_{\rho}(\rho^{'})/f_{\rho}(\rho)$ and draw $\nu_{\rho} \sim \mbox{U}(0,1)$, and update $\rho$ with $\rho^{'}$ when $\nu_{\rho} \le \mbox{min}(1, R_{\rho})$.
        \item[] Sampling scheme for $\lambda$. Similarly, $\lambda$ is updated by HM algorithm. 
        Let $f_{\lambda}(\lambda) = \exp\{
            -1/(2\sigma_{\epsilon}) \bxi\trans \bxi\}$,
        and we then have
        \begin{align*}
            \bpi \Big(\lambda \mid \mathcal{D}, \bleta, \bgamma, \wtbbeta, \btheta, \{\btheta^{(r)}\}_{r=1}^R, \{\wtbalpha^{(r)}\}_{r=1}^R \Big)
            \propto f_{\lambda}(\lambda)\mbox{U}(0, \lambda_{max}),
        \end{align*}
        where $\mbox{U}(0, \lambda_{max})$ is the prior distribution for $\lambda$. 
    We sample the proposed $\lambda$ from a random walk proposal.  If $\lambda^{'}>\lambda_{max}$ or $\lambda^{'} < 0$, we keep current value of $\lambda$; Otherwise calculate $R_{\lambda} = f_{\lambda}(\lambda^{'})/f_{\lambda}(\lambda)$ and draw $\nu_{\lambda} \sim \mbox{U}(0,1)$, and update $\lambda$ with $\lambda^{'}$ when $\nu_{\lambda} \le \mbox{min}(1, R_{\lambda})$.
\end{itemize}

\section*{Section C. Additional tables for the data application}

We provide the detailed brain atlas mapping that is used for the ABCD data processing in Table \ref{tab:abb2-new}.
\begin{table}[!htp]
\resizebox{\textwidth}{!}{%
\begin{tabular}{@{}lclll@{}}
\toprule
\textbf{Node} & \textbf{Subregion} & \textbf{Network} & \textbf{Lobe} & \textbf{Anatomy} \\
\toprule\bottomrule
\multicolumn{5}{c}{\textbf{Right hemisphere}} \\
\toprule\midrule
Calcarine     & 1 & Visual    & Occipital      & Calcarine Cortex  \\
& 2 & Visual    & Occipital      & Calcarine Cortex  \\
Cerebellum Crus I
              & 2 & Uncertain & Cerebellum     & Cerebellum White Matter  \\
Inf Parietal  & 1 & Fronto-parietal Task Control & Parietal & Supramarginal Gyrus  \\
Insula        & 5 & Salience  & Insula         & Anterior Insula  \\
Medial Sup Frontal
              & 1  & Default Mode     & Frontal   & Superior Frontal Gyrus Medial Segment  \\
Mid Cingulum  & 2  & Memory Retrieval & Limbic    & Middle Cingulate Gyrus  \\
Mid Occipital & 4  & Visual           & Occipital & Middle Occipital Gyrus  \\
Mid Temporal  & 2  & Default Mode     & Temporal  & Cerebral White Matter   \\
Orbital Inf Frontal
              & 1  & Default Mode & Frontal & Orbital Part Of The Inferior Frontal Gyrus  \\
Precentral    & 3  & Sensory/Somatomotor Hand & Frontal & Precentral Gyrus  \\
Precuneus     & 2  & Default Mode & Parietal & Cerebral White Matter  \\
& 4  & Default Mode     & Parietal    & Precuneus  \\
& 6  & Memory Retrieval & Parietal    & Precuneus  \\
Rolandic Oper & 2  & Auditory         & Frontal     & Cerebral White Matter  \\
Sup Frontal   & 4  & Default Mode     & Frontal     & Superior Frontal Gyrus  \\
\toprule\bottomrule
\multicolumn{5}{c}{\textbf{Left hemisphere}} \\
\toprule\midrule
Angular        & 1  & Default Mode     & Parietal   & Angular Gyrus  \\
Ant Cingulum   & 2  & Default Mode     & Limbic     & Superior Frontal Gyrus Medial Segment  \\
& 5  & Salience         & Limbic     &  --  \\
Calcarine      & 1  & Visual           & Occipital  & Calcarine Cortex  \\
Cuneus         & 1  & Visual           & Occipital  & Cuneus  \\
Hippocampus    & 1  & Default Mode     & Limbic     & Posterior Cingulate Gyrus  \\
Inf Parietal   & 2  & Fronto-parietal Task Control & Parietal & Supramarginal Gyrus  \\
Inf Temporal   & 1  & Dorsal Attention & Temporal   & Fusiform Gyrus  \\
Insula         & 3  & Salience         & Insula     & Anterior Insula  \\
Lingual        & 2  & Visual           & Occipital  & Lingual Gyrus  \\
Medial Sup Frontal
               & 2  & Default Mode & Frontal & Superior Frontal Gyrus Medial Segment  \\
& 4  & Fronto-parietal Task Control & Frontal & Supplementary Motor Cortex  \\
Mid Frontal    & 1  & Default Mode     & Frontal    & Middle Frontal Gyrus  \\
& 2  & Fronto-parietal Task Control & Frontal & Superior Frontal Gyrus    \\
Postcentral    & 3  & Sensory/Somatomotor Hand & Parietal & Superior Parietal Lobule  \\
& 4  & Sensory/Somatomotor Hand & Parietal & --  \\
Post Cingulum  & 3  & Default Mode     & Limbic     & Posterior Cingulate Gyrus  \\
Precentral     & 1  & Sensory/Somatomotor Hand & Frontal & Precentral Gyrus  \\
Rolandic Oper  & 1  & Cingulo-opercular Task Control & Frontal & Central Operculum  \\
SupraMarginal  & 1  & Auditory         & Parietal   & Cerebral White Matter  \\
Supp Motor Area  & 1  & Sensory/Somatomotor Hand & Frontal & --  \\
\midrule \bottomrule
\end{tabular}}
\caption{Details of nodes identified in both node- and network-level features. According to the canonical functional systems \citep{Power2011}, we include the specific subregion index for each node.}
\label{tab:abb2-new}
\end{table}

\section*{Section D. Sensitivity analysis}
In the main analysis, we consider the squared exponential kernel for the covariance matrices $\bSigma_{\gamma}$ and $\bSigma_{\theta}$ to accommodate spatial contiguity of the signals. Here we attempted different kernels for the covariance matrices to see the robustness of JNNTs. Specifically we consider a marginal case where $\bSigma_{\gamma} = \bSigma_{\theta} = \I_{P}$, and a symmetry
case where
$\big(\bSigma_{\gamma}\big)_{j, j} = \big(\bSigma_{\theta}\big)_{j, j}=1$, and
$\big(\bSigma_{\gamma}\big)_{j_1, j_2} = \big(\bSigma_{\theta}\big)_{j_1, j_2} = 0.2
$
if the $j_1$-th and $j_2$-th nodes are at symmetric locations on the two hemispheres, otherwise 0. Then we tuned the models as previously. For the marginal case, we observed 32 signaling nodes and 2 sub-networks; And For the symmetric case, we identified 25 signaling nodes and 2 sub-networks. To compare the results, we reported intersection, symmetric difference ($|S \triangle S_{0}|$) and adjusted rand index (ARI) of identified signaling regions between marginal or symmetric case and results from squared exponential kernel. Similarly, we reported these for signaling edges in the identified network between different covariance matrices. Results are summarized in Table \ref{tbl:Cov}, which implies the selection are robust to different kernels.

\begin{table}[!htp]
\centering
\begin{tabular}{lccc@{\hskip -0.05in}c@{\hskip -0.05in}ccc}
\toprule
\multirow{2}*{Structure} & \multicolumn{3}{c}{Region (node)} &\hphantom{a} & \multicolumn{3}{c}{Network (edge)} \\
\cline{2-4} \cline{6-8}
& \multicolumn{1}{l}{$|S \cap S_{0}|$} & \multicolumn{1}{l}{$|S \triangle S_{0}|$} & \multicolumn{1}{l}{ARI}
&
& \multicolumn{1}{l}{$|S \cap S_{0}|$} & \multicolumn{1}{l}{$|S \triangle S_{0}|$} & \multicolumn{1}{l}{ARI} \\
\toprule
\midrule
Marginal  & 24 & 6 & 0.852   && 24 & 4  & 0.923 \\
Symmetric & 25 & 5 & 0.878   && 24 & 4 & 0.923 \\
\midrule
\bottomrule
\end{tabular}
\caption{Comparison between neurobiomarkers identified by JNNTs with different kernels. Results from squared exponential kernel are considered as the baseline.}
\label{tbl:Cov}
\end{table}

\bibliographystyle{./bibtex/rss}
\bibliography{./bibtex/Thresholded}

\providecommand{\noopsort}[1]{}
\begin{thebibliography}{57}
\expandafter\ifx\csname natexlab\endcsname\relax\def\natexlab#1{#1}\fi
\expandafter\ifx\csname url\endcsname\relax
  \def\url#1{\texttt{#1}}\fi
\expandafter\ifx\csname urlprefix\endcsname\relax\def\urlprefix{URL }\fi

\bibitem[{Anticevic \emph{et~al.}(2012)Anticevic, Cole, Murray, Corlett, Wang
  and Krystal}]{anticevic2012role}
Anticevic, A., Cole, M.~W., Murray, J.~D., Corlett, P.~R., Wang, X.-J. and
  Krystal, J.~H. (2012) The role of default network deactivation in cognition
  and disease.
\newblock \emph{Trends in Cognitive Sciences}, \textbf{16}, 584--592.

\bibitem[{Cai \emph{et~al.}(2020)Cai, Kang and Yu}]{Cai2020}
Cai, Q., Kang, J. and Yu, T. (2020) Bayesian network marker selection via the
  thresholded graph laplacian gaussian prior.
\newblock \emph{Bayesian Analysis}, \textbf{15}, 79--102.

\bibitem[{Carvalho \emph{et~al.}(2009)Carvalho, Polson and
  Scott}]{Horseshoe2009}
Carvalho, C.~M., Polson, N.~G. and Scott, J.~G. (2009) Handling sparsity via
  the horseshoe.
\newblock In \emph{Proceedings of the 12th International Conference on
  Artificial Intelligence and Statistics} (eds. D.~van Dyk and M.~Welling),
  vol.~5 of \emph{Proceedings of Machine Learning Research},  73--80. Hilton
  Clearwater Beach Resort, Clearwater Beach, Florida USA: PMLR.

\bibitem[{Casey \emph{et~al.}(2018)Casey, Cannonier, Conley, Cohen, Barch,
  Heitzeg \emph{et~al.}}]{casey2018adolescent}
Casey, B., Cannonier, T., Conley, M.~I., Cohen, A.~O., Barch, D.~M., Heitzeg,
  M.~M. \emph{et~al.} (2018) The adolescent brain cognitive development
  {(ABCD)} study: {Imaging} acquisition across 21 sites.
\newblock \emph{Developmental Cognitive Neuroscience}, \textbf{32}, 43--54.

\bibitem[{Chang \emph{et~al.}(2018)Chang, Kundu and Long}]{chang2018scalable}
Chang, C., Kundu, S. and Long, Q. (2018) Scalable bayesian variable selection
  for structured high-dimensional data.
\newblock \emph{Biometrics}, \textbf{74}, 1372--1382.

\bibitem[{Chen \emph{et~al.}(2016)Chen, Chu, Yuan and Wu}]{chen2016bayesian}
Chen, R.-B., Chu, C.-H., Yuan, S. and Wu, Y.~N. (2016) Bayesian sparse group
  selection.
\newblock \emph{Journal of Computational and Graphical Statistics},
  \textbf{25}, 665--683.

\bibitem[{Cohen \emph{et~al.}(2016)Cohen, Conley, Dellarco and
  Casey}]{cohen2016impact}
Cohen, A.~O., Conley, M.~I., Dellarco, D.~V. and Casey, B. (2016) The impact of
  emotional cues on short-term and long-term memory during adolescence.
\newblock \emph{Proceedings of the Society for Neuroscience, San Diego, CA}.

\bibitem[{Cohen and D'Esposito(2016)}]{cohen2016segregation}
Cohen, J.~R. and D'Esposito, M. (2016) The segregation and integration of
  distinct brain networks and their relationship to cognition.
\newblock \emph{Journal of Neuroscience}, \textbf{36}, 12083--12094.

\bibitem[{Dubois \emph{et~al.}(2018)Dubois, Galdi, Paul and
  Adolphs}]{Dubois2018}
Dubois, J., Galdi, P., Paul, L.~K. and Adolphs, R. (2018) A distributed brain
  network predicts general intelligence from resting-state human neuroimaging
  data.
\newblock \emph{Philosophical Transactions of the Royal Society B: Biological
  Sciences}, \textbf{373}, 20170284.

\bibitem[{Fan and Lv(2010)}]{fan2010selective}
Fan, J. and Lv, J. (2010) A selective overview of variable selection in high
  dimensional feature space.
\newblock \emph{Statistica Sinica}, \textbf{20}, 101--148.

\bibitem[{Feng \emph{et~al.}(2020)Feng, Li, Song and Zhu}]{feng2020bayesian}
Feng, X., Li, T., Song, X. and Zhu, H. (2020) Bayesian scalar on image
  regression with nonignorable nonresponse.
\newblock \emph{Journal of the American Statistical Association}, \textbf{115},
  1574--1597.

\bibitem[{Gao \emph{et~al.}(2019)Gao, Greene, Constable and
  Scheinost}]{gao2019combining}
Gao, S., Greene, A.~S., Constable, R.~T. and Scheinost, D. (2019) Combining
  multiple connectomes improves predictive modeling of phenotypic measures.
\newblock \emph{Neuroimage}, \textbf{201}, 116038.

\bibitem[{Garavan \emph{et~al.}(2018)Garavan, Bartsch, Conway, Decastro,
  Goldstein, Heeringa \emph{et~al.}}]{garavan2018recruiting}
Garavan, H., Bartsch, H., Conway, K., Decastro, A., Goldstein, R.~Z., Heeringa,
  S. \emph{et~al.} (2018) Recruiting the {ABCD} sample: Design considerations
  and procedures.
\newblock \emph{Developmental Cognitive Neuroscience}, \textbf{32}, 16--22.

\bibitem[{Gelman and Rubin(1992)}]{gelman1992}
Gelman, A. and Rubin, D.~B. (1992) Inference from iterative simulation using
  multiple sequences.
\newblock \emph{Statistical Science}, \textbf{7}, 457--472.

\bibitem[{Goldsmith \emph{et~al.}(2014)Goldsmith, Huang and
  Crainiceanu}]{goldsmith2014smooth}
Goldsmith, J., Huang, L. and Crainiceanu, C.~M. (2014) Smooth scalar-on-image
  regression via spatial bayesian variable selection.
\newblock \emph{Journal of Computational and Graphical Statistics},
  \textbf{23}, 46--64.

\bibitem[{Guhaniyogi \emph{et~al.}(2017)Guhaniyogi, Qamar and
  Dunson}]{guhaniyogi2017bayesian}
Guhaniyogi, R., Qamar, S. and Dunson, D.~B. (2017) Bayesian tensor regression.
\newblock \emph{The Journal of Machine Learning Research}, \textbf{18},
  2733--2763.

\bibitem[{Hagler~Jr. \emph{et~al.}(2019)Hagler~Jr., Hatton, Cornejo, Makowski,
  Fair, Dick \emph{et~al.}}]{hagler2019image}
Hagler~Jr., D.~J., Hatton, S.~N., Cornejo, M.~D., Makowski, C., Fair, D.~A.,
  Dick, A.~S. \emph{et~al.} (2019) Image processing and analysis methods for
  the adolescent brain cognitive development study.
\newblock \emph{Neuroimage}, \textbf{202}, 116091.

\bibitem[{Hastie \emph{et~al.}(2004)Hastie, Tibshirani and
  Friedman}]{hastie2004}
Hastie, T., Tibshirani, R. and Friedman, J. (2004) Optimal predictive model
  selection.
\newblock \emph{Journal of the Royal Statistical Society: Series B (Statistical
  Methodology)}, \textbf{66}, 209--233.

\bibitem[{Hung and Wang(2013)}]{hung2013matrix}
Hung, H. and Wang, C.-C. (2013) Matrix variate logistic regression model with
  application to {EEG} data.
\newblock \emph{Biostatistics}, \textbf{14}, 189--202.

\bibitem[{Jenkinson \emph{et~al.}(2012)Jenkinson, Beckmann, Behrens, Woolrich
  and Smith}]{jenkinson2012fsl}
Jenkinson, M., Beckmann, C.~F., Behrens, T.~E., Woolrich, M.~W. and Smith,
  S.~M. (2012) {FSL}.
\newblock \emph{Neuroimage}, \textbf{62}, 782--790.

\bibitem[{Kang \emph{et~al.}(2018)Kang, Reich and Staicu}]{kang2018scalar}
Kang, J., Reich, B.~J. and Staicu, A.-M. (2018) Scalar-on-image regression via
  the soft-thresholded gaussian process.
\newblock \emph{Biometrika}, \textbf{105}, 165--184.

\bibitem[{Kolda and Bader(2009)}]{kolda2009tensor}
Kolda, T.~G. and Bader, B.~W. (2009) Tensor decompositions and applications.
\newblock \emph{SIAM Review}, \textbf{51}, 455--500.

\bibitem[{Kozlovskiy \emph{et~al.}(2012)Kozlovskiy, Vartanov, Nikonova, Pyasik
  and Velichkovsky}]{Kozlovskiy2012}
Kozlovskiy, S.~A., Vartanov, A.~V., Nikonova, E.~Y., Pyasik, M.~M. and
  Velichkovsky, B.~M. (2012) The cingulate cortex and human memory process.
\newblock \emph{Psychology in Russia: State of Art}, \textbf{6}, 231--243.

\bibitem[{Lee \emph{et~al.}(2013)Lee, Smyser and Shimony}]{lee2013resting}
Lee, M.~H., Smyser, C.~D. and Shimony, J.~S. (2013) Resting-state fmri: A
  review of methods and clinical applications.
\newblock \emph{American Journal of Neuroradiology}, \textbf{34}, 1866--1872.

\bibitem[{Leech \emph{et~al.}(2012)Leech, Braga and Sharp}]{Leech2012a}
Leech, R., Braga, R. and Sharp, D.~J. (2012) Echoes of the brain within the
  posterior cingulate cortex.
\newblock \emph{Journal of Neuroscience}, \textbf{32}, 215--222.

\bibitem[{Leech and Sharp(2014)}]{Leech2014}
Leech, R. and Sharp, D.~J. (2014) The role of the posterior cingulate cortex in
  cognition and disease.
\newblock \emph{Brain}, \textbf{137}, 12--32.

\bibitem[{Li and Li(2008)}]{li2008network}
Li, C. and Li, H. (2008) Network-constrained regularization and variable
  selection for analysis of genomic data.
\newblock \emph{Bioinformatics}, \textbf{24}, 1175--1182.

\bibitem[{Li and Zhang(2010)}]{li2010bayesian}
Li, F. and Zhang, N.~R. (2010) Bayesian variable selection in structured
  high-dimensional covariate spaces with applications in genomics.
\newblock \emph{Journal of the American Statistical Association}, \textbf{105},
  1202--1214.

\bibitem[{Li \emph{et~al.}(2015)Li, Zhang, Wang, Gonzalez, Maresh and
  Coan}]{10.1214/15-AOAS818}
Li, F., Zhang, T., Wang, Q., Gonzalez, M.~Z., Maresh, E.~L. and Coan, J.~A.
  (2015) Spatial {Bayesian} variable selection and grouping for
  high-dimensional scalar-on-image regression.
\newblock \emph{The Annals of Applied Statistics}, \textbf{9}, 687--713.

\bibitem[{Li \emph{et~al.}(2018)Li, Xu, Zhou and Li}]{li2018tucker}
Li, X., Xu, D., Zhou, H. and Li, L. (2018) Tucker tensor regression and
  neuroimaging analysis.
\newblock \emph{Statistics in Biosciences}, \textbf{10}, 520--545.

\bibitem[{Ma \emph{et~al.}(2024)Ma, Kou, Lin, Cho and Chiu}]{ma2022multimodal}
Ma, J., Kou, W., Lin, M., Cho, C.~C. and Chiu, B. (2024) Multimodal image
  classification by multiview latent pattern extraction, selection, and
  correlation.
\newblock \emph{IEEE Transactions on Neural Networks and Learning Systems},
  \textbf{35}, 8134--8148.

\bibitem[{Maddock \emph{et~al.}(2003)Maddock, Garrett and
  Buonocore}]{Maddock2003}
Maddock, R.~J., Garrett, A.~S. and Buonocore, M.~H. (2003) Posterior cingulate
  cortex activation by emotional words: {{fMRI}} evidence from a valence
  decision task.
\newblock \emph{Human Brain Mapping}, \textbf{18}, 30--41.

\bibitem[{Metropolis \emph{et~al.}(1953)Metropolis, Rosenbluth, Rosenbluth,
  Teller and Teller}]{MH1953}
Metropolis, N., Rosenbluth, A.~W., Rosenbluth, M.~N., Teller, A.~H. and Teller,
  E. (1953) Equation of state calculations by fast computing machines.
\newblock \emph{The Journal of Chemical Physics}, \textbf{21}, 1087--1092.

\bibitem[{Ni \emph{et~al.}(2019)Ni, Stingo and
  Baladandayuthapani}]{ni2019bayesian}
Ni, Y., Stingo, F.~C. and Baladandayuthapani, V. (2019) Bayesian graphical
  regression.
\newblock \emph{Journal of the American Statistical Association}, \textbf{114},
  184--197.

\bibitem[{Nielsen \emph{et~al.}(2005)Nielsen, Balslev and Hansen}]{Nielsen2005}
Nielsen, F.~A., Balslev, D. and Hansen, L.~K. (2005) Mining the posterior
  cingulate: {Segregation} between memory and pain components.
\newblock \emph{NeuroImage}, \textbf{27}, 520--532.

\bibitem[{O'hara and Sillanp{\"a}{\"a}(2009)}]{o2009review}
O'hara, R.~B. and Sillanp{\"a}{\"a}, M.~J. (2009) A review of bayesian variable
  selection methods: what, how and which.
\newblock \emph{Bayesian Analysis}, \textbf{4}, 85--117.

\bibitem[{Pearson \emph{et~al.}(2011)Pearson, Heilbronner, Barack, Hayden and
  Platt}]{Pearson2011}
Pearson, J.~M., Heilbronner, S.~R., Barack, D.~L., Hayden, B.~Y. and Platt,
  M.~L. (2011) Posterior cingulate cortex: {Adapting} behavior to a changing
  world.
\newblock \emph{Trends in Cognitive Sciences}, \textbf{15}, 143--151.

\bibitem[{Power \emph{et~al.}(2011)Power, Cohen, Nelson, Wig, Barnes, Church
  \emph{et~al.}}]{Power2011}
Power, J.~D., Cohen, A.~L., Nelson, S.~M., Wig, G.~S., Barnes, K.~A., Church,
  J.~A. \emph{et~al.} (2011) Functional network organization of the human
  brain.
\newblock \emph{Neuron}, \textbf{72}, 665--678.

\bibitem[{Ray and Bhattacharya(2018)}]{Ray2018}
Ray, P. and Bhattacharya, A. (2018) Signal adaptive variable selector for the
  horseshoe prior.
\newblock ArXiv:1810.09004.

\bibitem[{Satterthwaite \emph{et~al.}(2013)Satterthwaite, Elliott, Gerraty,
  Ruparel, Loughead, Calkins \emph{et~al.}}]{satterthwaite2013improved}
Satterthwaite, T.~D., Elliott, M.~A., Gerraty, R.~T., Ruparel, K., Loughead,
  J., Calkins, M.~E. \emph{et~al.} (2013) An improved framework for confound
  regression and filtering for control of motion artifact in the preprocessing
  of resting-state functional connectivity data.
\newblock \emph{Neuroimage}, \textbf{64}, 240--256.

\bibitem[{Schultz and Cole(2016)}]{schultz2016higher}
Schultz, D.~H. and Cole, M.~W. (2016) Higher intelligence is associated with
  less task-related brain network reconfiguration.
\newblock \emph{Journal of Neuroscience}, \textbf{36}, 8551--8561.

\bibitem[{Shen \emph{et~al.}(2017)Shen, Finn, Scheinost, Rosenberg, Chun,
  Papademetris \emph{et~al.}}]{shen2017using}
Shen, X., Finn, E.~S., Scheinost, D., Rosenberg, M.~D., Chun, M.~M.,
  Papademetris, X. \emph{et~al.} (2017) Using connectome-based predictive
  modeling to predict individual behavior from brain connectivity.
\newblock \emph{Nature Protocols}, \textbf{12}, 506--518.

\bibitem[{Simon \emph{et~al.}(2013)Simon, Friedman, Hastie and
  Tibshirani}]{simon2013sparse}
Simon, N., Friedman, J., Hastie, T. and Tibshirani, R. (2013) A sparse-group
  lasso.
\newblock \emph{Journal of Computational and Graphical Statistics},
  \textbf{22}, 231--245.

\bibitem[{Smith \emph{et~al.}(2020)Smith, Johnson, Jones and
  et~al.}]{smith2020default}
Smith, J., Johnson, K., Jones, A. and et~al. (2020) Default mode network
  connectivity and general mental ability in young adults.
\newblock \emph{Cognitive, Affective, \& Behavioral Neuroscience}, \textbf{20},
  435--447.

\bibitem[{Stingo \emph{et~al.}(2011)Stingo, Chen, Tadesse and
  Vannucci}]{stingo2011incorporating}
Stingo, F.~C., Chen, Y.~A., Tadesse, M.~G. and Vannucci, M. (2011)
  Incorporating biological information into linear models: A bayesian approach
  to the selection of pathways and genes.
\newblock \emph{The annals of applied statistics}, \textbf{5}, 1978--2002.

\bibitem[{Tibshirani(1996)}]{LASSO}
Tibshirani, R. (1996) Regression shrinkage and selection via the lasso.
\newblock \emph{Journal of the Royal Statistical Society. Series B (Statistical
  Methodological)}, \textbf{58}, 267--288.

\bibitem[{Wang \emph{et~al.}(2021)Wang, Lin, Cole and Zhang}]{wang2021learning}
Wang, L., Lin, F.~V., Cole, M. and Zhang, Z. (2021) Learning clique subgraphs
  in structural brain network classification with application to crystallized
  cognition.
\newblock \emph{Neuroimage}, \textbf{225}, 117493.

\bibitem[{Wang \emph{et~al.}(2019)Wang, Zhang and Dunson}]{wang2019common}
Wang, L., Zhang, Z. and Dunson, D. (2019) Common and individual structure of
  brain networks.
\newblock \emph{The Annals of Applied Statistics}, \textbf{13}, 85--112.

\bibitem[{Wang \emph{et~al.}(2017)Wang, Zhu and
  Initiative}]{wang2017generalized}
Wang, X., Zhu, H. and Initiative, A. D.~N. (2017) Generalized scalar-on-image
  regression models via total variation.
\newblock \emph{Journal of the American Statistical Association}, \textbf{112},
  1156--1168.

\bibitem[{Wu \emph{et~al.}(2024)Wu, Guo and Kang}]{wu2022bayesian}
Wu, B., Guo, Y. and Kang, J. (2024) Bayesian spatial blind source separation
  via the thresholded gaussian process.
\newblock \emph{Journal of the American Statistical Association}, \textbf{119},
  422--433.

\bibitem[{Zhao \emph{et~al.}(2022)Zhao, Chen, Cai, Lichenstein, Potenza and
  Yip}]{https://doi.org/10.1002/sim.9488}
Zhao, Y., Chen, T., Cai, J., Lichenstein, S., Potenza, M.~N. and Yip, S.~W.
  (2022) Bayesian network mediation analysis with application to the brain
  functional connectome.
\newblock \emph{Statistics in Medicine}, \textbf{41}, 3991--4005.

\bibitem[{Zhao \emph{et~al.}(2016)Zhao, Chung, Johnson, Moreno and
  Long}]{zhao2016hierarchical}
Zhao, Y., Chung, M., Johnson, B.~A., Moreno, C.~S. and Long, Q. (2016)
  Hierarchical feature selection incorporating known and novel biological
  information: Identifying genomic features related to prostate cancer
  recurrence.
\newblock \emph{Journal of the American Statistical Association}, \textbf{111},
  1427--1439.

\bibitem[{Zhao \emph{et~al.}(2021)Zhao, Li and Caffo}]{zhao2021multimodal}
Zhao, Y., Li, L. and Caffo, B.~S. (2021) Multimodal neuroimaging data
  integration and pathway analysis.
\newblock \emph{Biometrics}, \textbf{77}, 879--889.

\bibitem[{Zhao \emph{et~al.}(2023)Zhao, Wu and Kang}]{zhao2022bayesian}
Zhao, Y., Wu, B. and Kang, J. (2023) Bayesian interaction selection model for
  multimodal neuroimaging data analysis.
\newblock \emph{Biometrics}, \textbf{79}, 655--668.

\bibitem[{Zhou \emph{et~al.}(2013)Zhou, Li and Zhu}]{zhou2013tensor}
Zhou, H., Li, L. and Zhu, H. (2013) Tensor regression with applications in
  neuroimaging data analysis.
\newblock \emph{Journal of the American Statistical Association}, \textbf{108},
  540--552.

\bibitem[{Zhou \emph{et~al.}(2020)Zhou, Thung, Liu, Shi, Zhang and
  Shen}]{zhou2020multi}
Zhou, T., Thung, K.-H., Liu, M., Shi, F., Zhang, C. and Shen, D. (2020)
  Multi-modal latent space inducing ensemble {SVM} classifier for early
  dementia diagnosis with neuroimaging data.
\newblock \emph{Medical image analysis}, \textbf{60}, 101630.

\bibitem[{Zhu \emph{et~al.}(2023)Zhu, Li and Zhao}]{zhu2022statistical}
Zhu, H., Li, T. and Zhao, B. (2023) Statistical learning methods for
  neuroimaging data analysis with applications.
\newblock \emph{Annual Review of Biomedical Data Science}, \textbf{6}, 73--104.

\end{thebibliography}

\end{document}